\documentclass[a4paper, UKenglish, cleveref, autoref, thm-restate, numberwithinsect]{lipics-v2021}

\pdfoutput=1
\hideLIPIcs

\title{A Dichotomy Theorem for Homomorphism Problems on Automatic Structures}

\author{Antoine {Cuvelier}}{ENS Ulm, France}{}{}{}
\author{Rémi {Morvan}}{Université de Lille, France}{remi.morvan[at]univ-lille.fr}{https://orcid.org/0000-0002-1418-3405}{}

\authorrunning{A. Cuvelier and R. Morvan}

\Copyright{Antoine Cuvelier and Rémi Morvan}



\acknowledgements{
    We thank Pablo Barceló, Edgar Baucher, Diego Figueira and Joanna Fijalkow for helpful discussions.
    This paper is derived from part of the second author's Ph.D. thesis \cite{morvan2025thesis}, based on a work with the first author, while both were affiliated with LaBRI, Université de Bordeaux. The technical results are presented in \cite[Chapter~VIII]{morvan2025thesis}, while the introduction and preliminaries of this paper partially come from Chapters~I, II and VII.
}

\nolinenumbers 

\usepackage{tikz}
\usepackage{tikz-cd}

\usetikzlibrary{calc, positioning, arrows.meta, hobby, decorations.pathreplacing, decorations.markings, patterns, decorations.pathmorphing, fit, backgrounds, shapes, fadings}
\tikzfading[name=myfading, outer color=transparent!100, inner color=transparent!0]

\usetikzlibrary{automata}
\tikzset{
    initial text={},
    accepting/.style=accepting by double,
	every state/.style={minimum size=2em}
}

\tikzset{
	node distance = 1.5em,
	line width = .75pt,
	>={Classical TikZ Rightarrow},
	font = \footnotesize,
	big vertex/.style={
		draw,
		fill,
		fill opacity=.4,
		circle,
		line width=1pt,
		outer sep=2pt,
		inner sep=2pt
	},
	vertex/.style={
		draw,
		circle,
		line width=1pt,
		outer sep=2pt,
		inner sep=2pt
	},
	tiny vertex/.style={
		draw,
		circle,
		line width=.66pt,
		outer sep=0pt,
		inner sep=1.33pt
	},
	edge/.style={
		->,
		line width=1pt
	},
	implication/.style={
		->,
		double,
		line width=.5pt,
		arrows = {-Classical TikZ Rightarrow[length=0pt 3 .9]}
	}
}

\newrobustcmd{\drawHCGuess}[6]{
	\node[tiny vertex, #2=#3 of #1, draw=m0, fill=m0, fill opacity=#4] (#1-0) {};
	\node[tiny vertex, below=.1em of #1-0, draw=m1, fill=m1, fill opacity=#5] (#1-1) {};
	\node[tiny vertex, below=.1em of #1-1, draw=m2, fill=m2, fill opacity=#6] (#1-2) {};
}

\usetikzlibrary{fit}
\usetikzlibrary{arrows.meta}
\tikzset{
    inode/.style={
        inner xsep=0pt,
	}
}
\usepackage{xcolor}

\definecolor{mOrange}{HTML}{FF9800}
\definecolor{mPink}{HTML}{E91E63}
\definecolor{mPurple}{HTML}{9c27B0}
\definecolor{mBlue}{HTML}{2196F3}
\definecolor{mGreen}{HTML}{4CAF50}

\definecolor{mWhite}{HTML}{ffffff}
\definecolor{mLightGrey}{HTML}{cccccc}
\definecolor{mGrey}{HTML}{9e9e9e}
\definecolor{mDarkGrey}{HTML}{595959}
\definecolor{mBlack}{HTML}{2e2e2e}

\colorlet{maincolor}{mBlue!80!black}

\colorlet{KlDefn}{mOrange!75!black}
\colorlet{KlCite}{mBlue!40!black} 
\colorlet{KlLink}{mBlue!40!black}
\colorlet{KlUrl}{mBlue!40!black} 
\colorlet{KlWarning}{mPink!75!black}

\colorlet{m0}{mBlue} 
\colorlet{m1}{mPink}
\colorlet{m2}{mOrange}
\colorlet{m3}{mGreen}
\colorlet{m4}{mPurple}
\usepackage{stmaryrd}
\usepackage{mathcommand}
\usepackage[xcolor, hyperref, cleveref, notion, quotation, electronic]{knowledge}
\knowledgeconfigure{quotation, protect quotation={tikzcd}}
\knowledgeconfigure{diagnose line=true, diagnose bar=true}

\IfKnowledgePaperModeTF{
}{
    \knowledgestyle{intro notion}{color={KlDefn}, emphasize}
    \knowledgestyle{notion}{color={KlLink}}
    \hypersetup{
        colorlinks=true,
        breaklinks=true,
        linkcolor={KlLink}, 
        citecolor={KlCite}, 
        filecolor={KlLink}, 
        urlcolor={KlUrl}, 
    }
    \IfKnowledgeElectronicModeTF{
    }{
        \knowledgeconfigure{anchor point color={KlDefn}, anchor point shape=corner}
        \knowledgestyle{intro unknown}{color={KlWarning}, emphasize}
        \knowledgestyle{intro unknown cont}{color={KlWarning}, emphasize}
        \knowledgestyle{kl unknown}{color={KlWarning}}
        \knowledgestyle{kl unknown cont}{color={KlWarning}}
    }
}
\usepackage{mathtools} 
\usepackage{cancel} 
\usepackage{dutchcal} 

\usepackage[noend]{algorithm2e}
\SetKwInput{KwInput}{Input}
\makeatletter
\newcommand{\nosemic}{\renewcommand{\@endalgocfline}{\relax}}
\newcommand{\dosemic}{\renewcommand{\@endalgocfline}{\algocf@endline}}
\makeatother

\newrobustcmd{\case}[1]{\noindent{\color{lipicsYellow}\rule{0.73em}{0.73em}}~\textbf{\color{lipicsGray}#1}}
\newrobustcmd{\todo}[1]{{\color{mPink}\bfseries TODO: #1}}

\renewcommand{\emptyset}{\varnothing}
\newrobustcmd{\lBrack}{\llbracket}
\newrobustcmd{\rBrack}{\rrbracket}

\newcommand{\dcup}{\sqcup} 

\newlength{\dplabelwidth}
\newrobustcmd\decisionproblem[3]{%
	\begin{trivlist}\item[]\AP
	\hbox{%
		{\color{lipicsYellow}\vrule width .25em}\hspace{.9em}%
		\begin{minipage}{\dimexpr\linewidth-.25em-.9em\relax}
			\vspace{.15em}
			{#1}
			\begin{list}{}{%
				\setlength{\topsep}{.2em}\setlength{\partopsep}{0pt}%
				\setlength{\labelwidth}{\dplabelwidth}%
				\setlength{\labelsep}{.8em}%
				\setlength{\leftmargin}{\dimexpr\dplabelwidth+.8em\relax}%
				\setlength{\itemsep}{.15em}\setlength{\parsep}{0pt}%
				}
			\item[Input] #2
			\item[Question] #3
			\end{list}
			\vspace{.15em}
		\end{minipage}}%
	\end{trivlist}%
}
\newrobustcmd{\DPfont}[1]{\normalfont{\scshape #1}}

\newrobustcmd\smashxrightarrow[1]{%
  	\raisebox{-.04em}{$%
		\xrightarrow{\smash{\raisebox{-.1em}{%
	  		\tiny{#1}%
		}}}%
  	$}%
}%
\newrobustcmd\smashxleftarrow[1]{%
	\raisebox{-.04em}{$%
		\xleftarrow{\smash{\raisebox{-.1em}{%
			\tiny{#1}%
		}}}%
	$}%
}%
\newrobustcmd\smashxleftrightarrow[1]{%
	\raisebox{-.04em}{$%
		\xrightleftharpoons{\smash{\raisebox{-.1em}{%
			\tiny{#1}%
		}}}%
	$}%
}%

\newif\ifmainstatement
\mainstatementfalse

\newrobustcmd\introinrestatable[1]{%
	\ifmainstatement%
		\intro{#1}%
	\else%
		\kl{#1}%
	\fi%
}

\newrobustcmd\introinrestatableopt[1]{%
	\ifmainstatement%
		\intro[#1]{#1}%
	\else%
		\kl[#1]{#1}%
	\fi%
}

\newenvironment{mainstatement}{
  \mainstatementtrue
}{
  \mainstatementfalse
}
\newrobustcmd\defeq{\mathrel{\hat{=}}}
\newrobustcmd\?{\mathbf}
\newrobustcmd\+{\mathcal}

\newrobustcmd\card[1]{|1|}
\newcommand{\set}[1]{\{#1\}}
\newcommand{\tup}[1]{\langle#1\rangle}

\knowledgenewrobustcmd\pset[1]{\cmdkl{\mathfrak{P}(#1)}} 
\knowledgenewrobustcmd\psetp[1]{\cmdkl{\mathfrak{P}_+(#1)}} 

\knowledgenewrobustcmd{\transition}[1]{\mathrel{\cmdkl{\smashxrightarrow{\ensuremath{#1}}}}}

\knowledgenewrobustcmd{\semFO}[2]{\cmdkl{\lBrack} #1 \cmdkl{\rBrack^{#2}}} 
\knowledgenewrobustcmd\FOmodels{\mathrel{\cmdkl{\vDash}}} 
\newrobustcmd\notFOmodels{\mathrel{\kl[\FOmodels]{\nvDash}\;}}

\newrobustcmd\N{\mathbb{N}} 
\newrobustcmd\Np{\N_{>0}} 
\newrobustcmd\Z{\mathbb{Z}}
\newrobustcmd\Q{\mathbb{Q}}
\knowledgenewrobustcmd\intInt[1]{\cmdkl{\lBrack} #1 \cmdkl{\rBrack}} 
\usepackage{bbm}
\knowledgenewrobustcmd\2{\cmdkl{\mathbbm{2}}} 

\knowledgenewrobustcmd\restr[2]{{
  #1 \cmdkl{\raisebox{-.1em}{$\vert$}{}_{#2}}
}}

\knowledgenewrobustcmd{\adjacency}[4]{\cmdkl{\+{A\!d\!j}_{\!#2}^{\smash{#3,#4}}(#1)}}
\knowledgenewrobustcmd{\isom}{\mathrel{\cmdkl{\cong}}}

\knowledgenewrobustcmd\marked[1]{#1^{\cmdkl{\dag}}}
\knowledgenewrobustcmd\unarypred[1]{\cmdkl{P_{#1}}}
\knowledgenewrobustcmd\extendedSignature[2]{\cmdkl{#1_{#2}}}

\knowledgenewrobustcmd{\homto}{\mathrel{\cmdkl{\smashxrightarrow{\textsf{hom}}}}}
\newrobustcmd{\cohomto}{\mathrel{\kl[\homto]{\smashxleftarrow{\textsf{hom}}}}}
\newrobustcmd{\nothomto}{\mathrel{\kl[\homto]{\cancel{\smashxrightarrow{\textsf{hom}}}}}}
\newrobustcmd{\notcohomto}{\mathrel{\kl[\homto]{\cancel{\smashxleftarrow{\textsf{hom}}}}}}
\knowledgenewrobustcmd{\homregto}{\mathrel{\cmdkl{\smashxrightarrow{\textsf{reg hom}}}}}
\newrobustcmd{\nothomregto}{\mathrel{\kl[\homregto]{\cancel{\smashxrightarrow{\textsf{reg hom}}}}}}

\knowledgenewrobustcmd{\interpretation}[2]{\cmdkl{#1(#2)}}
\knowledgenewrobustcmd{\domainInter}[1]{\cmdkl{\mathrm{dom}_{#1}}}
\knowledgenewrobustcmd{\relInter}[2]{\cmdkl{#1_{\! #2}}}
\knowledgenewrobustcmd{\domainPres}[1]{\cmdkl{\mathrm{dom}_{#1}}}
\knowledgenewrobustcmd{\relPres}[2]{\cmdkl{#1_{\! #2}}}

\newrobustcmd{\smallpad}{%
	\hspace{.05em}\rule{.3em}{.075em}\hspace{.05em}
}
\knowledgenewrobustcmd{\pad}{\cmdkl{\smallpad}}
\knowledgenewrobustcmd{\convol}{\mathbin{\cmdkl{\otimes}}} 
\newrobustcmd\pair[2]{{\smaller\left(\begin{smallmatrix}%
	#1\\
	#2%
\end{smallmatrix}\right)}} 
\knowledgenewrobustcmd\SigmaPair[1][\Sigma]{\cmdkl{#1^{\smash{2}}_{{\otimes}}}}

\knowledgenewrobustcmd\core[1]{\cmdkl{\check{{#1}}}}
\knowledgenewrobustcmd\prodstruct{\mathbin{\cmdkl{\times}}}
\knowledgenewrobustcmd\powstruct[2]{\cmdkl{#1^{#2}}}
\knowledgenewrobustcmd\iterstruct[2]{#1^{\cmdkl{#2}}}
\knowledgenewrobustcmd\prodpres{\mathbin{\cmdkl{\underline{\times}}}}

\knowledgenewrobustcmd\Fin[1][\sigma]{\cmdkl{\mathrm{Fin}_{#1}}}
\knowledgenewrobustcmd\Aut[1][\sigma]{\cmdkl{\mathrm{Aut}_{#1}}}
\knowledgenewrobustcmd\FinPres[1][\sigma]{\cmdkl{\mathcal{F\!in}^{\smash{\textsf{pres}}}_{#1}}}
\knowledgenewrobustcmd\AutPres[1][\sigma]{\cmdkl{\mathcal{Aut}^{\smash{\textsf{pres}}}_{#1}}}

\newrobustcmd{\classStruct}{\mathcal{Cls}}

\newrobustcmd{\HomNoKl}[2]{\mathcal{H\!om}(#1,\,#2)}
\knowledgenewrobustcmd\Hom[2]{\cmdkl{\HomNoKl{#1}{#2}}}
\knowledgenewrobustcmd\HomFin[1]{\cmdkl{\HomNoKl{\mathrm{Fin}}{#1}}}
\knowledgenewrobustcmd\HomAut[1]{\cmdkl{\HomNoKl{\mathrm{Aut}}{#1}}}
\knowledgenewrobustcmd\HomAll[1]{\cmdkl{\HomNoKl{\mathrm{All}}{#1}}}

\newrobustcmd{\HomRegNoKl}[2]{\mathcal{H\!om}^{\smash{\textsf{reg}}}(#1,\,#2)}
\knowledgenewrobustcmd\HomReg[2]{\cmdkl{\HomRegNoKl{#1}{#2}}}
\knowledgenewrobustcmd\HomRegAut[1]{\cmdkl{\HomRegNoKl{\mathrm{Aut}}{#1}}}

\knowledgenewrobustcmd\link[1]{\cmdkl{\?L_{#1}}}
\knowledgenewrobustcmd{\zigzag}[2]{\cmdkl{\?Z_{#2}^{(#1)}}}
\knowledgenewrobustcmd{\transitiveTournament}[1]{\cmdkl{\?T_{#1}}}
\knowledgenewrobustcmd{\pathGraph}[1]{\cmdkl{\?P_{#1}}}
\knowledgenewrobustcmd{\clique}[1]{\cmdkl{\?K_{#1}}}
\newcommandPIE{\omegaClique}{\withkl{\kl[\omegaClique]}{\cmdkl{\?K_{<\omega}#1#2#3}}}
	\knowledge{\omegaClique}{notion}
\newcommandPIE{\omegaCliquePres}{\withkl{\kl[\omegaCliquePres]}{\cmdkl{\+K_{\!<\omega}#1#2#3}}}
	\knowledge{\omegaCliquePres}{notion}

\knowledgenewrobustcmd\projHom[1]{\cmdkl{\pi_{#1}}} 
\knowledgenewrobustcmd{\FederVardi}[1]{\cmdkl{\mathfrak{U}(#1)}} 
\knowledgenewrobustcmd\Myc{\mathop{\cmdkl{\mathfrak{M}}}} 
\knowledgenewrobustcmd\MycInf{\mathop{\cmdkl{\?M}}} 

\knowledgenewrobustcmd{\equalLength}{\mathrel{\cmdkl{\approx_{\textrm{len}}}}}
\knowledgenewrobustcmd{\prefix}{\mathrel{\cmdkl{\preccurlyeq_{\textrm{pref}}}}}
\knowledgenewrobustcmd{\lastLetter}[1]{\cmdkl{\+l_{#1}}}
\knowledgenewrobustcmd{\univStructSynchronous}[1]{\cmdkl{\?{#1^*}}}
\knowledgenewrobustcmd{\signatureSynchronous}[1]{\cmdkl{\sigma_{#1}^{\smash{\textrm{sync}}}}}

\knowledgenewrobustcmd{\subsumed}{\mathrel{\cmdkl{\sqsubseteq}}}
\knowledgenewrobustcmd{\LatticeGuessFunctions}[2]{\cmdkl{\langle \pset{#2}^{#1},\, \subsumed \rangle}}
\knowledgenewrobustcmd{\topLatticeGuessFunctions}[1]{\cmdkl{\Lambda_{#1}}}

\newcommandPIE{\HCOperator}{\withkl{\kl[\HCOperator]}{\cmdkl{\mathcal{H\!C}#1#2#3}}}
	\knowledge{\HCOperator}{notion}
\knowledgenewrobustcmd{\HCFixpoint}[2]{\cmdkl{H^{\,*}_{\smash{#1,#2}}}}

\knowledgenewrobustcmd{\unaryType}[2]{\cmdkl{\mu_{#2}(#1)}}
\knowledgenewrobustcmd{\structOfUnaryType}[1]{\cmdkl{\?1_{#1}}}
\knowledgenewrobustcmd{\ConstrUndecHom}[1]{#1^{\cmdkl{\star}}}
\knowledgenewrobustcmd{\rightquotient}[2]{#1\mathbin{\cmdkl{/}}#2}
\knowledgenewrobustcmd{\restrConnected}[1]{\cmdkl{\restr{#1}{\textrm{conn}}}}

\knowledgenewrobustcmd{\itemDTFinDual}{\cmdkl{\text{\normalfont{(\textsc{dt})$_\textsf{fin-dual}$}}}}
\knowledgenewrobustcmd{\itemDTHomDec}{\cmdkl{\text{\normalfont{(\textsc{dt})$_\textsf{hom-dec}$}}}}
\knowledgenewrobustcmd{\itemDTHomRegDec}{\cmdkl{\text{\normalfont{(\textsc{dt})$_\textsf{hom-reg-dec}$}}}}
\knowledgenewrobustcmd{\itemDTEqual}{\cmdkl{\text{\normalfont{(\textsc{dt})$_\textsf{equal}$}}}}
\knowledgenewrobustcmd{\itemDTFirstOrder}{\cmdkl{\text{\normalfont{(\textsc{dt})$_\textsf{first-order}$}}}}
\knowledgenewrobustcmd{\itemDTUnifFirstOrder}{\cmdkl{\text{\normalfont{(\textsc{dt})$_\textsf{unif-first-order}$}}}}

\knowledgenewrobustcmd{\RAT}{\cmdkl{\ensuremath{\textsc{Rat}}}}%
\knowledgenewrobustcmd{\DRAT}{\cmdkl{\ensuremath{\textsc{DRat}}}}%
\knowledgenewrobustcmd{\AUT}{\cmdkl{\ensuremath{\textsc{Aut}}}}%
\knowledgenewrobustcmd{\SYNC}{\cmdkl{\ensuremath{\textsc{Sync}}}}
\knowledgenewrobustcmd{\REC}{\cmdkl{\ensuremath{\textsc{Rec}}}}%
\knowledgenewrobustcmd\kREC[1][k]{\cmdkl{\ensuremath{#1\textsc{-Rec}}}}
\knowledgenewrobustcmd\kPROD[1][k]{\cmdkl{\ensuremath{#1\textsc{-Prod}}}}

\knowledgenewrobustcmd{\Id}{\cmdkl{\+{I\mkern-1.5mu d}}} 
\newrobustcmd{\bSymb}{\textcolor{mBlue}{b}}
\newrobustcmd{\rSymb}{\textcolor{mPink}{r}}

\input{kl/basic.kl}
\input{kl/complexity.kl}
\input{kl/automatic-structures.kl}
\input{kl/colouring.kl}
\input{kl/concrete-graphs.kl}
\input{kl/decision-problems.kl}
\input{kl/duality.kl}
\input{kl/first-order.kl}
\input{kl/homomorphism.kl}
\input{kl/misc.kl}
\input{kl/relations.kl}
\input{kl/structures.kl}
\input{kl/turing-machine.kl}

\begin{document}

\maketitle

\begin{abstract}
    The field of constraint satisfaction problems (CSPs) studies homomorphism problems between relational structures where the target structure is fixed. Classifying the complexity of these problems has been a central quest of the field, notably when both sides are finite structures. In this paper, we study the generalization where the input is an automatic structure—potentially infinite, but describable by finite automata.

    We prove a striking dichotomy: homomorphism problems over automatic structures are either decidable in non-deterministic logarithmic space (NL), or undecidable. We show that structures for which the problem is decidable are exactly those with finite duality, which is a classical property of target structures asserting that the existence of a homomorphism into them can be characterized by the absence of a finite set of obstructions in the source structure. Notably, this property precisely characterizes target structures whose homomorphism problem is definable in first-order logic, which is well-known to be decidable over automatic structures. 

    We also consider a natural variant of the problem. While automatic structures are finitely describable, homomorphisms from them into finite structures need not be. This leads to the notion of ``regular homomorphism'', where the homomorphism itself must be describable by finite automata. Remarkably, we prove that this variant exhibits the same dichotomy, with the same characterization for decidability.
\end{abstract}

\medskip
\case{Links.} This paper contains internal hyperlinks: clicking on a "notion@@notice"
leads to its \AP""definition@notion@notice"". 
Hyperlinks are generated using the   
"knowledge" package and its companion tool "knowledge-clustering".

\medskip
\case{AI Disclosure Statement.} Generative AI was not used in any way during the research process,
but was used for proofreading the paper.

\clearpage
\section{Introduction}
\label{sec:dichotomy-introduction}

\subparagraph*{Constraint Satisfaction Problems.}
The homomorphism problem takes two finite "relational structures" over the same "signature"%
---for instance two directed graphs---and asks whether there is a "homomorphism" from the former to the latter. This problem provides a natural framework to encode many other problems in computer 
science. For instance, whether a directed graph is "$k$-colourable" or whether it
does \textbf{not} admit any directed path of size $k$ can both be encoded as homomorphism problems. More involved examples include solving Sudoku grids and satisfiability of 3-SAT formulas.\footnote{This is folklore, see "eg" \cite[\S~I.3, pp.~39--42]{morvan2025thesis}.}
In all four cases, the input of the problem is encoded as the left-hand side of the homomorphism problem, while the constraints of the problem are encoded in the right-hand side.
Hence, the field of \emph{constraint satisfaction problems} studies the complexity of the 
homomorphism problem when its right-hand side, called the \AP""target structure"", is fixed.
The input to such a problem is called the \AP""source structure"",
and the problem can always be solved in "NP": guess a function, and then check that it is a "homomorphism". The examples of $3$-colouring or $3$-SAT prove that it can be "NP"-complete.
However, not all such problems are hard: the (non-)existence of a $k$-path in a directed graph is for instance solvable in polynomial time.

Classifying the complexity of these problems as a function of the "target structure" has hence been 
the central quest of the field, which culminated in 2017 when Bulatov
\cite[Theorem~1]{Bulatov2017DichotomyCSPs} and Zhuk \cite[Theorem~1.4]{Zhuk2020CSPDichotomy} 
independently proved the so-called ``dichotomy theorem'', conjectured 20 years earlier by Feder and 
Vardi \cite[\S~2, ``Dichotomy question'']{FederVardi1998ComputationalStructure}.
It states that every homomorphism problem (over finite structures and whose "target structure" is 
finite) is either in "P" or "NP"-complete.%
\footnote{In other words, there is no "NP"-intermediate problem among these problems.}
This result extends a result of Schaefer dating from 1978, who proved 
a similar statement for problems over the Boolean domain---meaning that the "target structure" can have only two elements \cite[Theorem~2.1]{Schaefer1978ComplexitySatisfiability}.

But beyond "P" and "NP", the "homomorphism" problem can exhibit
surprisingly low and varied complexities.
Even for Boolean domains,
Allender, Bauland, Immerman, Schnoor and Vollmer extended Schaefer's theorem to
prove that every not-so-easy problem---meaning that it is not solvable in "coNLogTime"---is
complete for one class among "NP", "P", "modL", "NL" or "L" under "Am0"-reductions
\cite[Theorem~3.1]{AllenderBaulandImmermanSchnoorVollmer2009Schaefer}.

For non-Boolean structures, the landscape becomes even more complex.
For instance, the non-existence of a path of size $k$ can be characterized by
the absence, in the "source structure", of a forbidden shape, called an "obstruction".
"Target structures" whose homomorphism problem shares this property are said to  
have "finite duality", and their homomorphism problem is always definable in first-order logic---which forms a strict subclass of problems solvable in deterministic logspace, see "eg" \cite[Theorems~3.1 \& 13.1, pp.~45 \& 203]{Immerman1998DescriptiveComplexity}.
This notion will play an important role in this paper.
Note that Atserias proved the converse implication to this result \cite[Corollary 4]{Atserias2008DigraphColoring}, see \Cref{prop:atserias},%
\footnote{This result was followed in the same year
by Rossman's theorem, which subsumes it \cite[Theorem~1.7]{Rossman2008Homomorphism}.}
and Larose and Tesson proved a (most useful) dichotomy theorem for "structures" with low complexity.
\AP\begin{proposition}[Larose--Tesson theorem {\cite[Theorem~3.2]{LaroseTesson2009UniversalAlgebraCSP}}]
	Let $\?B$ be a "finite $\sigma$-structure". If $\?B$ does not have "finite duality",
	then $\HomFin{\?B}$ is "L"-hard under "first-order reductions".
\end{proposition}

\subparagraph*{Automatic Structures.}
The results mentioned here all deal with \textbf{finite} structures; the extension of constraint satisfaction problems to infinite structures
is both an active field of research and a hard task. One of the most frequent
ways of finitely representing infinite structures is perhaps "via" "automatic structures", introduced by Hodgson in his Ph.D. thesis
\cite{Hodgson1976PhD}.\footnote{See "eg" \cite[\S~VII.3.1, p.~234]{morvan2025thesis} for the rich history of independent redefinitions of the notion.} 
While being potentially infinite, their representation via finite-state automata allows one
to decide model-checking of first-order logic over them \cite[Théorème~3.5]{Hodgson1983Decidabilite}. "Automatic structures" prove to be particularly interesting
given that they can actually encode various useful structures, including Presburger arithmetic
$\tup{\N, +}$---not only that, but the most common way of proving the decidability
of first-order logic over $\tup{\N, +}$ is probably by proving that $\tup{\N, +}$ is an "automatic structure"---, or the configuration graphs of all Turing machines.

The frontier of decidability over "automatic structures" is unfortunately quite blurry:
for instance the isomorphism problem is in general undecidable \cite[Theorem~5.15]{BlumensathGradel2004FinitePresentations}---and in fact it is complete for the first level
of the analytical hierarchy \cite[Theorem~5.9]{KhoussainovNiesRubinStephan2007Automatic}---but becomes decidable when restricted to ordinals \cite[Theorem~5.3]{KhoussainovRubinStephan2005AutomaticLinearOrders}, see also \cite{FinkelTodorcevic2013AutomaticOrdinals,JainKhoussainovSchlichtStephan2019IsomorphismTreeAutomaticOrdinals} for extensions of these results.
Despite automatic groups having been extensively studied since the late 1980s---see \cite{Rees2022AutomaticGroups} for a recent account---, the isomorphism problem over this class
remains wide open.

Kuske and Lohrey extensively studied problems on automatic graphs, including the existence of an Eulerian or a Hamiltonian path \cite{KuskeLohrey2010AutomaticGraphs}.
Concerning homomorphism problems, Köcher proved that whether an automatic graph is
"2-colourable"---or equivalently bipartite---is undecidable
\cite[Proposition~6.5]{Kocher2014AutomatischenGraphen}, and an analogous result was recently proved 
by Barceló, Figueira and Morvan for "$2$-regular colourability"
\cite[Theorem~4.4]{BarceloFigueiraMorvan2023SeparatingAutomatic}, which consists of asking
whether the graph is "2-colourable", with the extra constraint that the colouring itself
must be definable by an automaton.
For more background on "automatic structures", see the survey \cite{Gradel2020AutomaticStructures} or
Blumensath's monograph \cite[Chapter~XII]{Blumensath2024MSOModelTheory}.

\subparagraph*{Contributions: Homomorphism Problems on Automatic Structures.}

Our main result is a dichotomy theorem for "homomorphism problems" when the "target structure" is fixed and finite, and the source is any "automatic structure".

\begin{restatable*}[\introinrestatable{Dichotomy Theorem for Automatic Structures}]{theorem}{DichotomyThmDichotomyAutomatic}
	\AP\label{thm:dichotomy-theorem-automatic-structures}
	Let $\?B$ be a finite "$\sigma$-structure". The following are equivalent:
	\begin{description}
		\itemAP[\introinrestatable\itemDTFinDual.] $\?B$ has "finite duality";
		\itemAP[\introinrestatable\itemDTHomDec.] $\HomAut{\?B}$ is decidable;
		\itemAP[\introinrestatable\itemDTHomRegDec.] $\HomRegAut{\?B}$ is decidable;
		\itemAP[\introinrestatable\itemDTEqual.] $\HomAut{\?B} = \HomRegAut{\?B}$, "ie" for any "automatic presentation" $\+A$ of a 
		"$\sigma$-structure" $\?A$, there is a "homomorphism" from $\?A$ to $\?B$ "iff" 
		there is a "regular homomorphism" from $\+A$ to $\?B$;
		\itemAP[\introinrestatable\itemDTFirstOrder.] $\HomAll{\?B}$ is "first-order definable";
		\itemAP[\introinrestatable\itemDTUnifFirstOrder.] $\HomAll{\?B}$ has "uniformly first-order definable homomorphisms".\footnote{The notion of "uniformly first-order definable homomorphisms" is defined in \Cref{sec:uniformly-first-order-definable-hom}.}
	\end{description}
	Moreover, when $\HomAut{\?B}$ and $\HomRegAut{\?B}$ are undecidable, they are "coRE"-complete
	and "RE"-complete, respectively. When they are decidable, they are in "NL".
\end{restatable*}

In the statement above, "regular homomorphisms" (and the $\HomRegAut{-}$ notation) refer
to "homomorphisms" from the presentation of an "automatic structure" to a finite structure
that can themselves be defined by finite-state automata.

\subparagraph*{Organization.}
We provide in \Cref{sec:dichotomy-preliminaries} the preliminaries
of the paper, containing known results.
We move on in \Cref{sec:dichotomy-main-theorem} to both stating the
main result and providing an overview of the proof.
Decidability results are then shown in \Cref{sec:dichotomy-decidability}:
the decidability of $\HomAut{\?B}$ when $\?B$ has "finite duality"
is easy, and so most of the section is dedicated to the problem $\HomRegAut{\?B}$:
in \Cref{sec:first-order-implies-uniformity} we prove the decidability of $\HomRegAut{\?B}$, by showing that if the set of structures admitting a homomorphism to
a "structure" $\?B$ is "first-order definable", then actually
we can describe in first-order logic a uniform choice of such homomorphisms. 
Interestingly, while this implication assumes that the target structure has "finite duality", its proof crucially relies on a construction used
to handle "tree duality"---see \Cref{apdx:tree-duality}.

We then move to the undecidability results---one for "homomorphisms", and another one for "regular homomorphisms"---in \Cref{sec:dichotomy-undecidability}.
The undecidability proof for "homomorphisms" follows from an adaptation of
Larose and Tesson's "L"-hardness proof \cite[\S~3]{LaroseTesson2009UniversalAlgebraCSP},
and the fact that "Connectivity in Automatic Graphs" is undecidable.
The case of "regular homomorphisms" is similar but much more technical:
\begin{itemize}
	\item we start from the "regular reachability problem" for "linear Turing machines",
	which was proven to be undecidable by Barceló, Figueira and Morvan \cite[Lemma~4.2]{BarceloFigueiraMorvan2023SeparatingAutomatic};
	\item we reduce it to an intermediate problem called "Regular Unconnectivity in Automatic Graphs";
	\item and then, we reduce the latter problem 
		to $\HomRegAut{\?B}$ when $\?B$ does not have "finite duality",
		by adapting Larose and Tesson's proof.
\end{itemize}
We conclude in \Cref{sec:dichotomy-discussion} by extending our theorem to a larger class of 
structures, and discussing conjectures and related problems.

\subparagraph*{Related work: Regular Homomorphisms.}
This paper studies "regular homomorphisms" as a generalization of
the "regular $k$-colourings" of \cite{BarceloFigueiraMorvan2023SeparatingAutomatic}.
The related notion of ``regular isomorphism''
(under the name ``automatic isomorphism'') was studied in
\cite[Definition~6.10]{KhoussainovNerode1995AutomaticPresentations}. They showed that for any
"$\sigma$-structure" $\?A$,
there are either zero, one, or $\omega$ many "automatic presentations" of $\?A$
up to ``regular isomorphism'' \cite[Theorem~6.8]{KhoussainovNerode1995AutomaticPresentations}.

\subparagraph*{Related work: Homomorphism Problems.}
Let us mention first that other works deal with homomorphism problems on infinite structures. 
Freedman showed that a variant of 3-SAT was undecidable if formulas
were allowed to have atoms over pairs of integers \cite{Freedman1998KSat},
and this work was extended by Dantchev and Valencia \cite{DantchevValencia2005ComputationalLimits}.
In a similar setting, it was proven by Chen that 2-SAT and Horn-SAT are both decidable
\cite{Chen2005Periodic}. 
Klin, Kopczyński, Ochremiak and Toruńczyk then proved that
most homomorphism problems over a finite "target structure" 
become (only) exponentially harder if the "source structure" is allowed to be
a ``definable structure with atoms'' over the atoms $\tup{\Q, <}$
\cite[Theorem~3]{KlinKopczynskiOchremiakTorunczyk2015LocallyFiniteCSP}.
This was quickly followed by a result of Klin, Lasota, Ochremiak and Toruńczyk 
proving that when both structures are part of the input and both
are definable over the atoms $\tup{\N, =}$, then the existence of a homomorphism is
decidable, but the existence of a \emph{definable homomorphism} is not 
\cite[Theorem~3]{KlinLasotaOchremiakTorunczyk2016HomomorphismProblems}.
Note how this result contrasts with our dichotomy (\Cref{thm:dichotomy-theorem-automatic-structures}), since in our case the "homomorphism problem" and
the "regular homomorphism problem" are either both decidable or both undecidable.

Independently of and concurrently with this paper, Bulatov, Gong, Khoussainov and Wang
studied ``automatic CSPs'' \cite{bulatov2026automaticconstraintsatisfactionproblem}.
In our terms, their work studies "homomorphism problems" where
the signature is infinite (it has one predicate of arity $k$ for every $k \in \N$),
the source structure is "finite@@struct", and the target structure
has finite domain, while the interpretations of the predicates are given by
a single automaton. Their results are hence completely orthogonal to ours. 

We refer the reader to Bodirsky's survey \cite{Bodirsky2008Survey}---as
well as his more recent works---for "homomorphism problems" with infinite "target structures".
Lastly, let us point out that the brief description we gave of CSPs is incomplete
and viewed through the prism of computability/complexity.
Other work on the subject includes "eg" finding succinct representations
of "homomorphisms" when they exist \cite{BerkholzVinallSmeeth2023Dichotomy},
or the study of \emph{promise} CSPs
\cite{KrokhinOprsal2022PCSP}.
\section{Preliminaries}
\label{sec:dichotomy-preliminaries}

This section introduces the classical notions necessary to formalize \Cref{thm:dichotomy-theorem-automatic-structures}, namely "relational structures", 
"automatic structures" and "homomorphisms".
The notions introduced in this paper, mostly ``"regular homomorphisms"''
and ``"uniformly first-order definable homomorphisms"'',
will be defined in \Cref{sec:dichotomy-main-theorem}.

We let $\Sigma$ denote any finite alphabet and $\intro*\2$ be the alphabet $\{0,1\}$.
Given a set $X$, \AP$\intro*\pset{X}$ denotes its powerset and \AP$\intro*\psetp{X} \defeq \pset{X} \smallsetminus \set{\varnothing}$. The set \AP$\intro*\intInt{i,j}$, for $i,j \in \N$,
denotes the set of integers between $i$ and $j$ (inclusive). The restriction of a function $f$
to a subset $X$ of its domain is denoted by \AP$\intro*\restr{f}{X}$.

\subsection{Relational Structures}

\subparagraph*{Structures.}

"Relational structures" generalize the notion of directed graph by allowing (1) relations of higher arity and (2) multiple relations.

A \AP""relational signature"" consists of 
(1) a set of elements, called ""predicates"",
together with, for each of these elements, a strictly positive natural number, called its \emph{arity},
and (2) a set of so-called constant symbols.
Both sets can be empty or infinite.
We denote by $\+R_{(k)} \in \sigma$ the fact that "predicate" $\+R$, of arity
$k$, is part of a "signature" $\sigma$. 

Given a "signature" $\sigma$, a \AP""$\sigma$-structure"" $\?A$
then consists of
(1) a set $A$, called its \emph{domain},
(2) for each "predicate" $\+R_{(k)} \in \sigma$, a $k$-ary relation
	$\+R(\?A) \subseteq A^k$, and
(3) for each constant $c \in \sigma$, an element $c(\?A) \in A$.
We call $\+R(\?A)$ (resp. $c(\?A)$) the \AP""interpretation@@predicate""
of "predicate" $\+R_{(k)}$ (resp. constant $c$) in $\?A$.
By analogy with graphs, elements of the domain are sometimes referred to as
\emph{vertices}.

\begin{figure}
	\centering
	\begin{tikzpicture}
		\input{tikz/3-clique}
	\end{tikzpicture}
	\hspace{8em}
	\begin{tikzpicture}
		\input{tikz/2-transitive-tournament}
	\end{tikzpicture}
	\hspace{8em}
	\begin{tikzpicture}
		\input{tikz/2-path-graph}
	\end{tikzpicture}
	\caption{\label{fig:examples-graphs}From left to right: the "$3$-clique" $\clique{3}$, the "$2$-transitive tournament" $\transitiveTournament{2}$
	and the "$2$-path" $\pathGraph{2}$. Colours carry no information; they are only used
	to distinguish vertices in later figures.}
\end{figure}
Directed graphs can be seen as "$\sigma$-structures"
when $\sigma$ is the \AP""graph signature"", which contains a unique predicate $\+E_{(2)}$.
Let $k \in \N$. We define the \AP""$k$-clique"" $\intro*\clique{k}$ to be the
directed graph over $\intInt{1,k}$ with every edge $\tup{i,j}$ for $i \neq j$.
The \AP""$k$-transitive tournament"" $\intro*\transitiveTournament{k}$
and ""$k$-path"" $\intro*\pathGraph{k}$ both have $\intInt{0,k}$ as their domain,
and have an edge from $i$ to $j$ "iff" $i < j$ and $j = i + 1$, respectively.
See \Cref{fig:examples-graphs}.

The notions of ``"finite structure"'', ``"substructure"'', ``"proper substructure"'', ``"Cartesian product"'' and ``"induced substructure"''
are defined in the obvious manner: see \Cref{apdx:relational-structures}.

\subparagraph*{Homomorphisms.}

\begin{figure}
	\centering
	\begin{tikzpicture}
		\node[vertex, draw=m1, fill=m1, fill opacity=.4] at (0,1) (a) {};
		\node[vertex, draw=m0, fill=m0, fill opacity=.4] at (1.414,1) (b) {};
		\node[vertex, draw=m2, fill=m2, fill opacity=.4] at (2.828,1) (c) {};
		\node[vertex, draw=m1, fill=m1, fill opacity=.4] at (4.242,1) (d) {};
		\node[vertex, draw=m2, fill=m2, fill opacity=.4] at (0.707,0) (e) {};
		\node[vertex, draw=m1, fill=m1, fill opacity=.4] at (2.121,0) (f) {};
		\node[vertex, draw=m3, fill=m3, fill opacity=.4] at (3.536,0) (g) {};
		\node[vertex, draw=m0, fill=m0, fill opacity=.4] at (4.950,0) (h) {};

		\draw[edge] (a) to (b);
		\draw[edge] (b) to (e);
		\draw[edge] (f) to (b);
		\draw[edge] (b) to (c);
		\draw[edge] (f) to (c);
		\draw[edge] (g) to (f);
		\draw[edge] (c) to (g);
		\draw[edge] (d) to (c);
		\draw[edge] (g) to (d);
		\draw[edge] (d) to (h);

		\begin{scope}[xshift=5.5cm]
			\node[vertex, draw=m0, fill=m0, fill opacity=.4] at (1.414,1) (b2) {};
			\node[vertex, draw=m2, fill=m2, fill opacity=.4] at (2.828,1) (m2) {};
			\node[vertex, draw=m1, fill=m1, fill opacity=.4] at (2.121,0) (f2) {};
			\node[vertex, draw=m3, fill=m3, fill opacity=.4] at (3.536,0) (g2) {};

			\draw[edge] (f2) to (b2);
			\draw[edge] (b2) to (m2);
			\draw[edge] (f2) to (m2);
			\draw[edge] (g2) to (f2);
			\draw[edge] (m2) to (g2);
		\end{scope}
	\end{tikzpicture}
	\caption{
		\AP\label{fig:prelim-hom}
		Two graphs, and a "homomorphism" from the left graph
		to the right graph.
		The colours are not part of the "structure", but
		are used to describe the "homomorphism":
		every vertex of some colour of the left-hand side is mapped
		onto the unique vertex having this colour on the right-hand side.
	}
\end{figure}
\AP""Homomorphisms"" from a "$\sigma$-structure" $\?A$
to a "$\sigma$-structure" $\?B$ sharing the same "signature"
are defined as usual: they correspond to functions
preserving the relational information, "ie" the image of the interpretation
of any constant in $\?A$ must be the interpretation of the same constant in $\?B$,
and moreover if a tuple belongs to the interpretation of a
relation in $\?A$, its image must belong to the interpretation of the same
relation in $\?B$. On graphs, this corresponds to the classical
notion of (directed) graph homomorphisms, see \Cref{fig:prelim-hom}
for an example.
Variations of the notion of "homomorphisms", namely "isomorphisms",
"automorphisms", "embeddings" and "strong-onto homomorphisms" are defined
in the obvious way in \Cref{apdx:relational-structures}.
We denote by \AP$\?A \intro*\homto \?B$ the existence of a "homomorphism" from $\?A$ to $\?B$.

\subparagraph*{First-Order Logic.}
We assume that the reader is familiar with \AP""first-order logic"" 
over "relational structures":
recall that its syntax is given by the grammar
\[
	\phi \mathrel{{:}{=}} 
	\+R_{(k)}(x_1, \hdots, x_k)
	\mid x = y
	\mid \phi \land \phi
	\mid \phi \lor \phi
	\mid \neg \phi
	\mid \exists x.\; \phi
	\mid \forall x.\; \phi,
\]
where $k$ ranges over $\N$ and $\+R_{(k)}$ ranges over $k$-ary predicates of the "signature".
The semantics is defined by an obvious induction.
We will denote by $\phi(x_1,\hdots,x_n)$
the pair of a formula $\phi$ together with a tuple of variables
$(x_1,\hdots,x_n)$ spanning over exactly all free variables of $\phi$---repetitions are allowed.
We denote the semantics of $\phi(x_1,\hdots,x_n)$ over a
"$\sigma$-structure" $\?A$ by \AP$\intro*\semFO{\phi(x_1,\hdots,x_n)}{\?A}$.
By definition, it is a subset of $A^n$, and we denote by
\AP$\?A, \tup{a_1,\hdots,a_n} \intro*\FOmodels \phi(x_1,\hdots,x_n)$
the fact that a tuple of points
$\tup{a_1,\hdots,a_n}$ belongs to the set $\semFO{\phi(x_1,\hdots,x_n)}{\?A}$.

A class of "finite@@struct" (resp. arbitrary) "$\sigma$-structures" equipped with a tuple of its elements---that we will call \emph{pointed structures}---is said to be \AP""first-order definable"" whenever it equals a class of the form 
$\bigl\{
	\tup{\?A,\, a_1,\, \hdots,\, a_n} \mid
	\tup{a_1,\hdots,a_n} \in \semFO{\phi(x_1,\hdots,x_n)}{\?A}
\bigr\}$
for some $n \in \N$ and some "first-order formula" $\phi(x_1,\hdots,x_n)$. In the class comprehension, $\?A$ implicitly ranges
over all "finite@@struct" (resp. arbitrary) "structures". 

A \AP""first-order reduction"" between two classes of
pointed finite (resp. arbitrary) "structures" 
is defined as a function from all finite (resp. arbitrary)
"structures" to itself, that can be described in "first-order logic",
and that preserves membership in the classes. See "eg"
\cite[Defs.~2.11 \& 1.26]{Immerman1998DescriptiveComplexity}
for a formal definition.
Finally, two classes are \AP""first-order equivalent"" whenever
they are first-order equireducible.

\subsection{Automatic Relations and Structures}

\subparagraph*{Automatic Relations.}
Let $k\in \Np$. Given words $u_1,\hdots,u_k \in \Sigma^*$,
we let $u_1 \intro*\convol u_2 \reintro*\convol \hdots \reintro*\convol u_k$ denote the word
of size $\max(|u_1|,\hdots,|u_k|)$ over $(\Sigma\cup{\pad})^k \smallsetminus \set{(\pad, \hdots, \pad)}$ obtained by writing each $u_i$ on a left-aligned horizontal tape,
adding padding symbols $\pad$ at the end of the shorter words,
and then reading $k$-tuples of letters from left to right.
Here, the padding symbol $\intro*\pad$ denotes a fixed symbol lying outside $\Sigma$,
and $\Sigma\cup{\pad}$ denotes the alphabet $\Sigma\cup\{\pad\}$.

\begin{figure}
	\centering
	\begin{tikzpicture}[shorten >= 1pt, node distance = 1.8cm, on grid, baseline]
		\node[state, initial left, accepting] (q0) {}; 
		\path[->]
			(q0) edge[loop above] node[font=\scriptsize] {$\pair{a}{a}, \pair{\pad}{a}\text{ for }a\in \Sigma$} (q0);
	\end{tikzpicture}
	\hspace{7em}
	\begin{tikzpicture}[shorten >= 1pt, node distance = 1.8cm, on grid, baseline]
		\node[state, initial left] (q0) {}; 
		\node[state, accepting, right=of q0] (q1) {}; 
		\path[->]
			(q0) edge[loop above] node[font=\scriptsize] {$\pair{a}{a}\text{ for }a\in \Sigma$} (q0);
		\path[->]
			(q0) edge node[font=\scriptsize, below=1em] {$\pair{\pad}{a}\text{ for }a\in \Sigma$} (q1);
	\end{tikzpicture}
	\caption{
		\AP\label{fig:example-synchronous-automata}
		"Synchronous automata" "recognizing@@syncauto"
		the prefix relation (left)
		and $\set{\tup{u,ua} \mid u \in \Sigma^* \land a \in \Sigma}$ (right).
	}
\end{figure}

A (finite-state) \AP""synchronous automaton"" $\+A$ of arity $k \in \Np$ over $\Sigma$
is a finite-state automaton over $(\Sigma\cup{\pad})^k \smallsetminus \set{(\pad, \hdots, \pad)}$.
We say that $\+A$ accepts the $k$-tuple of words $\tup{u_1,\hdots,u_k}$ whenever
it accepts $u_1 \convol \hdots \convol u_k$.
It \AP""recognizes@@syncauto"" a relation $\+R \subseteq (\Sigma^*)^k$ when
$\+R$ is precisely the set of $k$-tuples of words it accepts.
Examples of such automata are given in \Cref{fig:example-synchronous-automata}.
A relation is \AP""automatic@@rel"" when it is "recognized@@syncauto" by some "synchronous automaton".

\subparagraph*{A Model-Theoretic Perspective on Automatic Relations.}
We define on $\Sigma^*$:
\begin{itemize}
	\itemAP a unary relation $\intro*\lastLetter{a}$ indicating that the last letter of a word is $a$,
	\itemAP a binary relation $\intro*\equalLength$ indicating that two words have the same length,
	\itemAP a binary relation $\intro*\prefix$ indicating that a word is a prefix of another.
\end{itemize} 
We denote by $\intro*\signatureSynchronous{\Sigma}$ the "signature" $\langle \langle\lastLetter{a}\rangle_{a \in \Sigma},\, \equalLength,\, \prefix \rangle$%
\footnote{For the sake of simplicity, we abusively use the same notations for
the "predicates" and their "interpretations@@predicate" in the structure $\univStructSynchronous{\Sigma}$ defined below.} and
by \AP$\intro*\univStructSynchronous{\Sigma}$ the "$\signatureSynchronous{\Sigma}$-structure"
over $\Sigma^*$ where the "predicates" are "interpreted@@predicate" as above.

\begin{proposition}[{\cite[Theorems~1 \& 2]{EilenbergElgotShepherdson1969TapeAutomata}}]
	\AP\label{prop:automatic-first-order}
	If $\Sigma$ has at least two letters, then a relation over $\Sigma^*$
	is "automatic@@rel" "iff" it is "first-order definable" in
	$\univStructSynchronous{\Sigma}$.\footnote{When $\Sigma$ has a single letter,
	then the right-to-left implication holds, but not the converse one.}
\end{proposition}

\subparagraph*{Automatic Structures.}
We fix a finite "relational signature" $\sigma$.
An \AP""automatic presentation $\+A$ of a $\sigma$-structure"" consists of
(1) an alphabet $\Sigma$,
(2) a regular language $\intro*\domainPres{\+A} \subseteq \Sigma^*$, and
(3) for every predicate $\+R_{(k)} \in \sigma$, an
"automatic relation" $\intro*\relPres{\+R}{\+A} \subseteq (\Sigma^*)^k$.

The \AP\reintro{structure represented} by an "automatic presentation" $\+A$, denoted by $\?A$, has
$\domainPres{\+A}$ as its domain, and the predicate $\+R_{(k)} \in \sigma$
is "interpreted@@pred" as $\relPres{\+R}{\+A}$.
Given $u \in \domainPres{\+A}$, we will sometimes denote by $\+A(u)$ instead of $u$
the element of the "structure represented" by $\+A$ that is represented by $u$.
This notation only serves to emphasize the semantic nature of the objects: we
will think of $u$ as a finite word but of $\+A(u)$ as an element of a "structure".

We say that a "$\sigma$-structure" is \AP""automatic@@struct"" if it is
"represented@@struct" by an "automatic presentation". For instance,
it is not hard to see that the infinite binary tree $\?B$ is "automatic@@struct":
it is "represented@@struct" by the "automatic presentation" $\+B$ over the alphabet $\2$ with $\domainPres{\+B} \defeq \2^*$ and
$\relPres{\+E}{\+B} \defeq \{\tup{u,ua} \mid u \in \2^* \land a \in \2\}$.
A "synchronous automaton" for $\relPres{\+E}{\+B}$ is represented on the
right-hand side of \Cref{fig:example-synchronous-automata}.

\begin{proposition}[{Variation of \cite[Théorème~3.5]{Hodgson1983Decidabilite}}]
	\AP\label{prop:data-complexity-model-checking}
	\label{prop:first-order-model-checking-automatic-structures}
	The problem of, given a first-order formula $\phi$ over $\sigma$
	and an "automatic presentation" $\+A$ of a "$\sigma$-structure" $\?A$,
	to decide whether $\?A$ satisfies $\phi$, is decidable.
	If moreover $\phi$ is fixed and of the form
	$\exists x_1.\,\dotsc \exists x_k.\, \psi(x_1,\dotsc,x_k)$,
	where $\psi$ contains neither quantifiers nor negations, then 
	the problem becomes "NL"-complete.
\end{proposition}

In \Cref{apdx:construction-automatic-presentations}, we show that the classical
constructions ("eg" "Cartesian product") on "relational structures"
can actually be implemented with "automatic presentations".
Let us mention that we will make heavy use
of \Cref{prop:automatic-first-order} to logically manipulate "automatic structures".
For an alternative logic-based approach to "automatic structures", we refer the reader
to \cite{ColcombetLoding2007WMSO}.

\subsection{Cores and Duality}

\subparagraph*{Cores.}
This paper studies "homomorphisms" between structures.
Note that if two structures $\?A$ and $\?A'$ are such that $\?A \homto \?A'$
and $\?A' \homto \?A$, then studying the existence of homomorphisms from (resp. to) $\?A$
or from (resp. to) $\?A'$ makes no difference. In this case, we say that $\?A$ and $\?A'$
are \AP""homomorphically equivalent"".
An example of such a pair is given in \Cref{fig:prelim-hom}: a "homomorphism" from the
larger "structure" to the smaller one is given by the colours, and the other way around is trivial since the small structure is a "substructure" of the bigger one.

It so happens that each equivalence class of "homomorphically equivalent" finite structures
admits a unique minimal structure, called its \AP""core"". We denote by \AP$\intro*\core{\?A}$ the
"core" of a "structure" $\?A$. It can actually be characterized as follows:
it is the unique (up to "isomorphism") minimal "substructure" $\?C$ of $\?A$ such that there is a "homomorphism" $r\colon \?A \to \?C$ with the property that $r(c) = c$ for each $c\in C$.
In \Cref{fig:prelim-hom}, the right-hand side "structure"
is actually the "core" of the left-hand side graph.
We call \reintro{core} any "$\sigma$-structure" that is the "core" of
some "structure"---or equivalently of itself.

\subparagraph*{Idempotent Cores.}%
We fix a "relational signature" $\sigma$.
Given a "$\sigma$-structure" $\?B$,
we denote by \AP$\intro*\extendedSignature{\sigma}{\?B}$
the signature obtained from $\sigma$ by adding
a unary predicate \AP$\intro*\unarypred{b}$ for each $b\in B$.
The \AP""marked structure"" \AP$\intro*\marked{\?B}$ of $\?B$ is the
"$\extendedSignature{\sigma}{\?B}$-structure"
obtained from $\?B$ by "interpreting@@predicate" each predicate $\unarypred{b}$ as the
singleton $\{b\}$.

The "marked structure" $\marked{\core{\?B}}$ of the "core" of $\?B$ is
usually called the \emph{idempotent core}
of $\?B$.
The following reduction is considered folklore, see "eg"
\cite[Lemma 2.5]{LaroseTesson2009UniversalAlgebraCSP},
but we propose in \Cref{apdx-prop:idempotent-core-preserves-csp-complexity}
a self-contained and illustrated proof.
This reduction is a central tool
in the algebraic approach to understanding CSPs.
\begin{restatable}{proposition}{idempotentCore}
	\AP\label{prop:idempotent-core-preserves-csp-complexity}
	If $\?B$ is a finite "core", then the problems $\HomAll{\marked{\?B}}$ and 
	$\HomAll{\?B}$ are "first-order equivalent".
	Moreover, this equivalence preserves "finiteness@@structure",
	in the sense that "finite structures" are mapped to "finite structures".  
	Hence, by restricting this equivalence, we also obtain that
	$\HomFin{\marked{\?B}}$ and $\HomFin{\?B}$ are "first-order equivalent".
\end{restatable}
From this proposition, we get that $\HomFin{\?B}$ and
$\HomFin{\marked{\core{\?B}}}$ are "first-order equivalent".

\subsection{Obstructions and Finite Duality}

Let $\?B$ and $\?D$ be "finite $\sigma$-structures".
We say that $\?D$ is an \AP""obstruction"" of $\?B$ when $\?D \nothomto \?B$.
In this case, note that finding $\?D$ inside $\?A$---in the sense that $\?D \homto \?A$---implies that $\?A$ cannot have a "homomorphism" to $\?B$: in this sense, the presence of
$\?D$ is an \emph{obstruction} to the existence of a "homomorphism" to $\?B$.

A \AP""dual"" of $\?B$ is any set $\+D$ of "finite $\sigma$-structures" "st"
for any "finite $\sigma$-structure" $\?A$:
$
	\?A \homto \?B
	\text{ "iff" }
	\forall \?D \in \+D,\, \?D \nothomto \?A.
$
Note that any "dual" must only contain "obstructions" of $\?B$.%
\footnote{The notion of dual was implicitly introduced
by Ne\v{s}et\v{r}il and Pultr, who proved that for undirected graphs, no
non-trivial "finite duality" existed, in the sense that the
only "cores" $\?B$ having a finite "dual" were the empty graph, and the graph having
a single vertex and no edge \cite[Corollary 4.1]{NesetrilPultr1978Duality}.}
The set of all "obstructions" of $\?B$ is a "dual" of $\?B$.

\begin{example}
	\AP\label{ex:dual-T2} Let $k\in\N$. The set
	$\{\pathGraph{k+1}\}$ is a "dual" of the
	"$k$-transitive tournament" $\transitiveTournament{k}$.
	
	Indeed, $\pathGraph{k+1} \nothomto \transitiveTournament{k}$. Conversely,
	if $\?G$ is a "finite graph" "st" $\pathGraph{k+1} \nothomto \?G$, then
	mapping each vertex of $\?G$ to the length of the longest directed path ending in it
	defines a "homomorphism" from $\?G$ to $\transitiveTournament{k}$.
\end{example}

Moreover, if $\+D$ and $\+D'$ are sets of "obstructions" of $\?B$ and $\+D \subseteq \+D'$,
then if $\+D$ is a "dual", so is $\+D'$: hence, the goal is to find \emph{small} duals.
For this reason, "duals" are also called \emph{complete sets of obstructions}.
We say that $\?B$ has \AP""finite duality"" if it admits a finite "dual", "ie"
consisting of finitely many "structures". For instance, $\transitiveTournament{k}$ has "finite duality".

An "obstruction" $\?D$ of $\?B$ is \AP""critical@@obs""\,\footnote{We borrow the terminology
from \cite{LaroseLotenTardif2007CharacterisationFOCSP}.} when for every
"proper substructure" $\?D'$ of $\?D$, we have $\?D' \homto \?B$.
Clearly, every "critical obstruction" must be a "core".

Note first that the set of all "critical obstructions" of $\?B$ is a "dual" of $\?B$.
Indeed, if $\?A \nothomto \?B$, then $\?A$ is an "obstruction" and so it must contain---by well-foundedness of $\N$---a "critical obstruction" $\?D$ as a "substructure".

\begin{proposition}
	\label{prop:finite-duality-iff-critical-obstructions}
	Let $\?B$ be a "finite $\sigma$-structure". $\?B$ has "finite duality"
	"iff" it has finitely many "critical obstructions".
\end{proposition}

\begin{proof}
	The right-to-left implication is trivial. For the converse one,
	let $\+D = \{\?D_1,\dotsc,\?D_m\}$ be a finite "dual" of $\?B$.
	If $\?C$ is a "critical obstruction" of $\?B$ then in particular $\?C \nothomto \?B$
	and so there exists $i \in \intInt{1,m}$ "st" there is a "homomorphism" $f$ 
	from $\?D_i$ to $\?C$. Now the image of $f$ is again an "obstruction" of $\?B$,
	and since $\?C$ is "critical@@obs", it follows that it must be $\?C$ itself. In other
	words, $f$ is "strong-onto".
	Hence, there are only finitely many "critical obstructions" of $\?B$.
\end{proof}

A key use of \Cref{prop:finite-duality-iff-critical-obstructions} is
to prove that some structures do not have "finite duality".

\begin{figure}
	\centering 
	\begin{tikzpicture}
		\input{tikz/zigzag-graph-hom-T2}
	\end{tikzpicture}
	\caption{\AP\label{fig:zigzag-graph-hom-T2}A "homomorphism" from the "zigzag graph" (left-hand side) to the "$2$-transitive tournament" (right-hand side).}
\end{figure}

\begin{example}[{\Cref{ex:dual-T2}, continued.}]
	\AP\label{ex:zigzag-defn}
	Let $n\in\N$.
	We define the \AP""zigzag graph"" $\intro*\zigzag{n}{2}$ of width $n$ and length 2
	to be the graph whose vertices are $a_0, \dotsc, a_n$, $b_0, \dotsc, b_{n}$,
	$a'_0$ and $b'_n$, with edges from $a_i$ to $b_{i-1}$ and to $b_{i}$ (for $i \in \intInt{0,n}$, whenever the nodes exist), and with an edge from $a'_0$ to $a_0$ and from $b_n$
	to $b'_n$. See the left-hand side of \Cref{fig:zigzag-graph-hom-T2} for an illustration.
	
	Note that $\zigzag{n}{2}$ does not admit a "homomorphism" to the "$2$-path"---indeed, such a homomorphism should send $a'_0$, $a_0$ and $b_0$ onto $0$, $1$, and $2$, respectively, 
	and so all $a_i$'s (resp. $b_i$'s) must be sent onto $1$ (resp. $2$), but then $b'_n$ cannot be mapped anywhere. We actually prove in \Cref{apdx:critical-obst-2path}
	that each $\zigzag{n}{2}$ ($n\in\N$) is a "critical obstruction" of $\pathGraph{2}$.

	On the other hand, each $\zigzag{n}{2}$ with $n\in\Np$ admits a "homomorphism" to the "$2$-transitive tournament", as witnessed by \Cref{fig:zigzag-graph-hom-T2}.%
	\footnote{However, observe that $\zigzag{0}{2} = \pathGraph{3}$ is an "obstruction" of
	$\transitiveTournament{2}$.}
	In fact, this "homomorphism" is far from being unique:
	each vertex $a_1,\,a_2,\,\dotsc,\,a_{n-1}$ can be sent to either $0$ or $1$
	(the blue and red vertices), 
	and similarly, each vertex $b_1,\,b_2,\,\dotsc,\,b_{n-1}$ can be sent to either $1$ or $2$
	(the red and yellow vertices).
\end{example}

\begin{proposition}[""Atserias' theorem"" {\cite{Atserias2008DigraphColoring}}]
	\label{prop:atserias}
	\AP Let $\?B$ be a "finite $\sigma$-structure". Then $\?B$ has "finite duality"
	if, and only if, $\HomAll{\?B}$ is "first-order definable".
\end{proposition}

\Cref{ex:zigzag-defn} provides some intuition on this theorem.
Since $\transitiveTournament{2}$ has "finite duality",
$\HomAll{\transitiveTournament{2}}$ can be defined by the first-order formula saying
``the "structure" does not contain $\pathGraph{3}$''.
However, this idea cannot be adapted to
define $\HomAll{\pathGraph{2}}$
since ``excluding a $\zigzag{n}{2}$ graph for some $n$'' is not
"first-order definable"---which follows from a locality argument.
\section{The Dichotomy Theorem}
\label{sec:dichotomy-main-theorem}

This section introduces
the last technical notions we require ("regular homomorphisms"
and "uniformly first-order definable homomorphisms"), and then states our main result, namely \Cref{thm:dichotomy-theorem-automatic-structures}.

\subsection{Regular Homomorphisms and Homomorphism Problems}

A \AP""regular function"" from a regular language $K \subseteq \Sigma^*$ into a finite set
$L$ is any function $f\colon K \to L$ such that $f^{-1}[v]$ is a regular language for each $v \in L$.
A more general definition, which also covers infinite $L$ and which we will need, is given in \Cref{apdx:regular-homomorphisms}.

A \AP""regular homomorphism"" from a "presentation@@auto" $\+A$ to
a "presentation@@auto" $\+B$ of a \textbf{finite} structure $\?B$
is a "regular function" that is also
a "homomorphism" from $\?A$ to $\?B$.
Since $\?B$ is finite, the existence of a "regular homomorphism" to $\?B$
is actually independent of its "presentation@@auto", and we will hence write $\+A \reintro*\homregto \?B$ 
to mean that $\+A \homregto \+B$.

Given a fixed "signature" $\sigma$ and a $\sigma$-structure $\?B$,
we denote by:
\begin{itemize}
	\itemAP $\intro*\HomFin{\?B}$ (resp. $\intro*\HomAut{\?B}$, resp.
	$\intro*\HomAll{\?B}$) the class of all "finite $\sigma$-structures"
	(resp. "automatic $\sigma$-structures", resp. arbitrary "$\sigma$-structures")
	that admit a "homomorphism" to $\?B$,
	\itemAP $\intro*\HomRegAut{\?B}$ the class of all "automatic presentations of  $\sigma$-structures" that admit a "regular homomorphism" to $\?B$.
\end{itemize}
Somewhat abusively, we identify these classes with the associated decision problems---except
for $\HomAll{\?B}$ since arbitrary "$\sigma$-structures" cannot be encoded as finite strings.
For finite structures, we assume the input to be given using adjacency lists, and for "automatic 
structures", we assume the input to be described by an "automatic presentation".
We call $\HomAut{\?B}$ the \AP""homomorphism problem"" over $\?B$ and
$\HomRegAut{\?B}$ the \AP""regular homomorphism problem"" over $\?B$.

The classical notion of graph colourability can be formalized as a homomorphism problem:
let $k\in \N$, then a \AP""$k$-colouring"" of a graph is simply a "homomorphism"
from it to $\clique{k}$.
Similarly, we define a \AP""regular $k$-colouring"" of the "presentation of an automatic graph"
to be a "regular homomorphism" from it to $\clique{k}$.

The results of Köcher \cite[Proposition~6.5]{Kocher2014AutomatischenGraphen} and Barceló, Figueira and Morvan \cite[Theorem~4.4]{BarceloFigueiraMorvan2023SeparatingAutomatic} state
that it is undecidable whether (the presentation of) an "automatic graph" is "$2$-colourable"
or "$2$-regularly colourable", respectively.
This suggests that most "homomorphism" and "regular homomorphism problems"
on "automatic structures" are undecidable,
which \Cref{thm:dichotomy-theorem-automatic-structures} makes precise.

\subsection{Uniformly First-Order Definable Homomorphisms}
\label{sec:uniformly-first-order-definable-hom}

\begin{definition}
	\AP We say that $\HomFin{\?B}$ (resp. $\HomAll{\?B}$)
	has \AP""uniformly first-order definable homomorphisms"" if there exist "first-order formulas"
	$\langle \phi_b(x) \rangle_{b\in B}$ over $\sigma$ "st" for any finite
	(resp. arbitrary) "$\sigma$-structure" $\?A$, for any $a\in A$,
	there is at most one $b \in B$, denoted by $b(a)$
	"st" $\langle \?A, a \rangle \FOmodels \phi_b(x)$, and moreover, if $\?A \homto \?B$
	then $a \mapsto b(a)$ is a "homomorphism" from $\?A$ to $\?B$.
\end{definition}

The adverb ``uniformly'' in ``"uniformly first-order definable homomorphisms"''
refers to the fact that the formulas do not depend on the "source structure".
This property says that not only is the class of all structures mapping to $\?B$
"definable in first-order logic", but a witness of this property,
namely a "homomorphism", can also be described in this logic.

\begin{example}[{\Cref{ex:zigzag-defn}, continued}]
	\AP\label{ex:zigzag-FO}
	For instance, $\HomAll{\transitiveTournament{2}}$ has "uniformly first-order definable homomorphisms", by letting $\phi_0(x)$ define the set of vertices with no
	predecessor, $\phi_2(x)$ define the set of vertices with no successor (but at least one predecessor), and $\phi_1(x)$ define the set of vertices that satisfy neither $\phi_0$ nor $\phi_2$.

	On the other hand, looking at the "zigzag graph" of \Cref{ex:zigzag-defn} and \Cref{fig:zigzag-graph-hom-T2}
	for long enough will convince the reader that no such strategy can work
	for $\pathGraph{2}$.
\end{example}

\subsection{Overview of the Dichotomy Theorem}
\label{sec:dichotomy-overview}

\begin{mainstatement}
	\DichotomyThmDichotomyAutomatic
\end{mainstatement}

\begin{figure}
	\centering
	\begin{tikzpicture}
		\input{tikz/dichotomy-overview}
	\end{tikzpicture}
	\caption{\AP\label{fig:dichotomy-overview}Implications shown in the paper to prove
	\Cref{thm:dichotomy-theorem-automatic-structures}.}
\end{figure}
We prove \Cref{thm:dichotomy-theorem-automatic-structures} by showing the
implications depicted in \Cref{fig:dichotomy-overview}.
The most difficult implications are $\itemDTFinDual \Rightarrow \itemDTUnifFirstOrder$,
which we prove in \Cref{sec:first-order-implies-uniformity}, 
$\itemDTHomDec \Rightarrow \itemDTFinDual$ and $\itemDTHomRegDec \Rightarrow \itemDTFinDual$,
which we prove by contraposition in \Cref{sec:undecidability-hom,sec:undecidability-homreg}.

On the other hand, the implications $\itemDTUnifFirstOrder \Rightarrow \itemDTEqual$,
$\itemDTEqual \Rightarrow \itemDTHomRegDec$ and $\itemDTEqual \Rightarrow \itemDTHomDec$ are straightforward, and we prove them in \Cref{sec:dichotomy-decidability}, together with $\itemDTFinDual \Rightarrow \itemDTHomDec$ as a warm-up.
\section{Decidability of the Homomorphism Problems}
\label{sec:dichotomy-decidability}

\subsection{The Easy Implications}

\subparagraph*{Decidability of the Homomorphism Problem.}

\begin{proposition}%
	\AP\label{prop:dichotomy-FinDual-implies-HomDec}
	Let $\?B$ be a "finite $\sigma$-structure".
	If $\?B$ has "finite duality", then $\HomAut{\?B}$ is decidable in "NL".%
	\footnote{This corresponds to the implication $\itemDTFinDual \Rightarrow \itemDTHomDec$
	of \Cref{thm:dichotomy-theorem-automatic-structures}.}
\end{proposition}

\begin{proof}
	Given a "finite $\sigma$-structure" $\?D$ with domain $\{d_1,\dotsc,d_n\}$,
	we build the "first-order formula"
	\[
		\phi_{\?D} \defeq \exists x_1.\; \cdots \exists x_n.\;
		\bigwedge_{\+R_{(k)} \in \sigma} \bigwedge_{\substack{\tup{i_1,\dotsc,i_k} \in \intInt{1,n}^k\\\text{"st" }\tup{d_{i_1},\dotsc,d_{i_k}} \in \+R(\?D)}}
		\+R(x_{i_1},\dotsc,x_{i_k}).
	\]
	By construction, for any "$\sigma$-structure" $\?A$, we have $\?A \FOmodels \phi_{\?D}$
	"iff" $\?D \homto \?A$.
	Then, since $\?B$ has "finite duality", it admits a finite "dual"
	$\?D_1,\dotsc,\?D_m$.
	From \Cref{prop:de-bruijn-erdos} it follows that 
	$\?A \homto \?B$ "iff" $\?D_i \nothomto \?A$ for each $i$, and so 
	$
		\?A \FOmodels \bigwedge_{i=1}^m \neg \phi_{\?D_i}
		\quad\text{"iff"}\quad 
		\?A \homto \?B.
	$
	Since $\bigvee_{i=1}^m \phi_{\?D_i}$ is equivalent to a formula of the form of \Cref{prop:first-order-model-checking-automatic-structures},
	the conclusion follows from this proposition, together with the fact that "NL" = "coNL".
\end{proof}

\subparagraph*{Equality of the Homomorphism Problems Implies their Decidability.}

\begin{proposition}%
	\AP\label{prop:dichotomy-Equal-implies-HomDec}
	Let $\?B$ be a "finite $\sigma$-structure".
	If $\HomAut{\?B} = \HomRegAut{\?B}$ then $\HomAut{\?B}$ and $\HomRegAut{\?B}$ are decidable.%
	\footnote{This corresponds to the implications $\itemDTEqual \Rightarrow \itemDTHomDec$
	and $\itemDTEqual \Rightarrow \itemDTHomRegDec$
	of \Cref{thm:dichotomy-theorem-automatic-structures}.}
\end{proposition}

To prove this, we first give an upper bound on the homomorphism problems independently
of any assumption on $\?B$.
\begin{restatable}{proposition}{dichotomyGeneralUpperBounds}
	\AP\label{prop:dichotomy-general-upper-bounds}
	Let $\?B$ be a "finite $\sigma$-structure".
	Then $\HomAut{\?B}$ is "coRE" and $\HomRegAut{\?B}$ is "RE".
\end{restatable}

\begin{proof}[Proof sketch]
	For the formal proof, see \Cref{apdx-prop:dichotomy-general-upper-bounds}.

	For $\HomAut{\?B}$, we can enumerate "finite substructures" of the "source structure",
	and, when we find one that does not map to $\?B$, answer `no'.

	On the other hand, for $\HomRegAut{\?B}$ we can simply enumerate $|B|$-tuples of finite automata and check that they do define a "regular homomorphism" to $\?B$, in which case we answer `yes'. 
\end{proof}

\begin{proof}[Proof of \Cref{prop:dichotomy-Equal-implies-HomDec}.]
	If $\HomAut{\?B} = \HomRegAut{\?B}$, then by \Cref{prop:dichotomy-general-upper-bounds},
	these problems are both "RE" and "coRE", and are hence decidable.
\end{proof}

\subsection{First-Order Definability Implies Uniformity}
\label{sec:first-order-implies-uniformity}

\begin{restatable}{lemma}{finiteDualityUniformlyDefinableHomomorphisms}
	\AP\label{lemma:finite-duality-uniformly-definable-homomorphisms}
	Let $\?B$ be a finite structure. Then $\HomAll{\?B}$ is "first-order definable" "iff"
	$\HomAll{\?B}$ has "uniformly first-order definable homomorphisms".%
	\footnote{The same 
	equivalence holds if one replaces $\HomAll{\?B}$ with $\HomFin{\?B}$.
	This corresponds to the implication $\itemDTFirstOrder \Rightarrow \itemDTUnifFirstOrder$
	of \Cref{thm:dichotomy-theorem-automatic-structures} and its converse implication.}
\end{restatable}

\begin{proof}[Proof sketch]
	For the formal proof, see \Cref{apdx-lemma:finite-duality-uniformly-definable-homomorphisms}.
	The converse implication is trivial.

	For any structure $\?A$, we consider the function $F\colon A \to \pset{B}$ which maps
	$a$ to the set of $b \in B$ "st" at least one "homomorphism" from $\?A$ to $\?B$
	sends $a$ to $b$.
	When $\?A \homto \?B$, the image of $F$ actually consists of non-empty sets.
	Under this assumption, we show that $F$ actually defines a "homomorphism"
	from $\?A$ to $\FederVardi{\?B}$, which is a "relational structure"
	introduced by Feder and Vardi whose domain is $\psetp{B}$.
	Feder and Vardi showed that this structure is such that
	$\FederVardi{\?B} \homto \?B$ under the assumption that
	$\HomAll{\?B}$ is "first-order definable".%
	\footnote{Actually this is true for a superclass of structures, namely
	those with "tree duality", and the property $\FederVardi{\?B} \homto \?B$
	exactly characterizes these structures.}
	Overall, this yields a "homomorphism" from $\?A$ to $\?B$.
	We use the assumption that $\HomAll{\?B}$ is "first-order definable"
	together with the "marked structure" $\marked{\?B}$,
	to show that, for each $b$, each preimage $\set{a \in A \mid b\in F(a)}$
	is "definable in first-order logic". Using the fact that 
	any fixed "homomorphism" from $\FederVardi{\?B}$ to $\?B$
	can be described in first-order logic, for the trivial reason that
	both sides are "finite@@struct", we conclude by compositionality
	that the "homomorphism" built from $\?A$ to $\?B$ ("via" $\FederVardi{\?B}$)
	can be described by "first-order formulas" that actually do not depend on $\?A$.
\end{proof}

"Uniformly first-order definable homomorphisms" are actually a very strong restriction:
we show that such "homomorphisms" are always "regular@@hom".
\begin{proposition}%
	\AP\label{prop:uniformly-first-order-implies-regular}
	Let $\?B$ be a finite "$\sigma$-structure".
	If $\HomAll{\?B}$ has "uniformly first-order definable homomorphisms",
	then $\HomRegAut{\?B} = \HomAut{\?B}$.%
	\footnote{This corresponds to the implication $\itemDTUnifFirstOrder \Rightarrow \itemDTEqual$
	of \Cref{thm:dichotomy-theorem-automatic-structures}.}
\end{proposition}

\begin{proof}
	Let $\+A$ be an "automatic presentation" of a "$\sigma$-structure" $\?A$,
	and assume that $\?A \homto \?B$. We need to show that $\+A \homregto \?B$.
	Let $\langle \phi_b(x) \rangle_{b\in B}$ be "first-order formulas" over $\sigma$
	as in the definition of "uniformly first-order definable homomorphisms".
	
	Since $\+A$ is an "automatic presentation" over $\Sigma$---where we assume "wlog" that $\Sigma$
	has at least two letters, so that \Cref{prop:automatic-first-order} applies---,
	for each "predicate" $\+R$ of arity $k$
	of $\sigma$, there exists a "first-order formula" $\psi_{\+R}(x_1,\dotsc,x_k)$ over 
	$\signatureSynchronous{\Sigma}$ describing each relation $\+R$,
	and a "first-order formula" $\psi_{\textsf{dom}}$ describing $\domainPres{\+A}$.
	We then define $\phi^*_b(x)$ as the formula obtained from $\phi_b(x)$
	by (1) substituting $\psi_{\+R}(x_1,\dotsc,x_k)$ for $\+R(x_1,\dotsc,x_k)$,
	and (2) relativizing the quantifiers "wrt" $\psi_{\textsf{dom}}$.%
	\footnote{This second operation consists of recursively
	transforming formulas $\exists x.\; \chi(x)$ into $\exists x.\; \psi_{\textsf{dom}}(x) \land \chi(x)$,
	and formulas $\forall x.\; \chi(x)$ into $\forall x.\; \psi_{\textsf{dom}}(x) \Rightarrow \chi(x)$.}

	Then, for each $b\in B$, $\phi^*_b(x)$ is a "first-order formula" over $\signatureSynchronous{\Sigma}$,
	and so $\{u \in \Sigma^* \mid \langle \univStructSynchronous{\Sigma},u \rangle \FOmodels \phi^*_b(x)\}$ is regular by \Cref{prop:automatic-first-order}.
	Clearly, these sets are disjoint and cover $\domainPres{\+A}$, and the function that
	maps $u \in \domainPres{\+A}$ to the unique $b$ "st" $\langle \univStructSynchronous{\Sigma},u \rangle \FOmodels \phi^*_b(x)$ is a "homomorphism".
	Hence, we have built a "regular homomorphism" from $\+A$ to $\?B$, which concludes the proof.
\end{proof}

Putting together \Cref{lemma:finite-duality-uniformly-definable-homomorphisms}, \Cref{prop:uniformly-first-order-implies-regular} and "Atserias' theorem", we get that if $\?B$ has "finite duality",
then $\HomRegAut{\?B} = \HomAut{\?B}$.
In turn, since $\HomAut{\?B}$ is "coRE" and $\HomRegAut{\?B}$ is "RE"
(\Cref{prop:dichotomy-general-upper-bounds}), this implies
that $\HomRegAut{\?B}$ is decidable.
In fact, $\+A \homregto \?B$ can even be expressed by a "first-order formula" over $\univStructSynchronous{\Sigma}$, see \Cref{apdx:first-order-formula-homreg}.
\section{Undecidability of the Homomorphism Problems}
\label{sec:dichotomy-undecidability}

We now prove the undecidability of $\HomAut{\?B}$ and $\HomRegAut{\?B}$
when $\?B$ does not have "finite duality". Both reductions are
adapted from Larose and Tesson's proof that $\HomFin{\?B}$ is "L"-hard when $\?B$ does not
have "finite duality" \cite[Theorem 3.2]{LaroseTesson2009UniversalAlgebraCSP}.
However, proving the undecidability of the problem that is reduced
to $\HomRegAut{\?B}$ is not as easy as for $\HomAut{\?B}$.
\begin{itemize}
	\item For $\HomAut{\?B}$, we reduce the complement of "Connectivity in Automatic Graphs",
		providing a "coRE" lower bound.
	\item For $\HomRegAut{\?B}$, we reduce "Regular Unconnectivity in Automatic Graphs",
		which in turn is reduced from the "regular reachability problem".
\end{itemize}

\subsection{Preliminaries: Power Structures}

This subsection recalls the classical notion of ``power structure'',
which is crucial in establishing the undecidability of our two problems.

Given two "structures" $\?A$ and $\?B$, we define the structure \AP$\intro*\powstruct{\?B}{\?A}$ as follows: its domain is the set of all maps from $A$ to $B$,
and for every "predicate" $\+R$ of arity $k$, for any maps $f_1,\dotsc,f_k$,
we have $\langle f_1,\dotsc,f_k\rangle \in \+R(\powstruct{\?B}{\?A})$ when 
for every $\langle a_1,\dotsc,a_k\rangle \in \+R(\?A)$, we have 
$\langle f_1(a_1), \dotsc, f_k(a_k) \rangle \in \+R(\?B)$. 
The following statement is proven in \Cref{apdx-prop:currying-hom}.

\begin{restatable}[Folklore: Currying Homomorphisms]{proposition}{curryingHom}
	\AP\label{prop:currying-hom}
	Given "structures" $\?A$, $\?B$ and $\?C$, if $f\colon \?A\prodstruct \?B \to \?C$
	is a "homomorphism", then $F\colon \?A \to \powstruct{\?C}{\?B}$,
	defined by $a \mapsto (b \mapsto f(a,b))$, is a "homomorphism".
	In fact, this mapping $f \mapsto F$ is a bijection
	between "homomorphisms" $\?A\prodstruct \?B \to \?C$
	and "homomorphisms" $\?A \to \powstruct{\?C}{\?B}$.
\end{restatable}

\subsection{Undecidability of \,$\HomAut{\?B}$}
\label{sec:undecidability-hom}

\begin{restatable}{lemma}{lowerboundHom}
	\AP\label{lem:lowerbound-hom}
	If $\?B$ does not have "finite duality", then $\HomAut{\?B}$
	is "coRE"-hard.%
	\footnote{In the case of Larose and Tesson, they study the problem
	$\HomFin{-}$, and prove in \cite[Theorem 3.2]{LaroseTesson2009UniversalAlgebraCSP}
	that there is a "first-order reduction" from "Connectivity in Finite Graphs" to
	$\HomFin{\marked{\?B}}$ for any $\?B$ that does not have "finite duality".
	Together with \Cref{prop:idempotent-core-preserves-csp-complexity},
	this shows that $\HomFin{\?B}$ is "L"-hard under "first-order reductions".}%
\end{restatable}

\begin{proof}[Proof sketch]
	See the formal proof in \Cref{apdx-lem:lowerbound-hom}.
	Assuming that $\?B$ does not have "finite duality", we prove first
	that its "signature" must contain at least one predicate of arity at least 2
	(\Cref{prop:finite-duality-unary-predicates}).
	From there, we reduce the \textbf{complement} of the following problem
	\AP
	\decisionproblem{""Connectivity in Automatic Graphs""}{
		An "automatic presentation" $\+G$ of a directed graph,
		and two elements $s,t \in \Sigma^*$.
	}{
		Are $\+G(s)$ and $\+G(t)$ "connected" in $\?G$?
	}
	to $\HomAut{\?B}$ (\Cref{lem:reduction-hom}), by a straightforward adaptation of Larose and Tesson's reduction for the finite case. The main idea of their reduction is that
	the fact that $\?B$ does not have "finite duality" can be described as a non-connectivity-like property, more precisely as a property of the form ``$\projHom{1}$ and $\projHom{2}$ are not connected(ish) in a structure of the form $\powstruct{\?B}{\?X}$''
	(\Cref{prop:characterization-finite-duality-path-projections}).
	From there, we reduce an instance $\tup{\+G, s, t}$ of "Connectivity in Automatic Graphs"
	to an instance looking like $\?G' \prodstruct \?X \homto^? \?B$ where $\?G'$ is essentially
	$\?G$ encoded as a "$\sigma$-structure".\footnote{It is actually denoted by $\?A$ in the formal proof.} By currying, the property
	$\?G' \prodstruct \?X \homto^? \?B$ rephrases as $\?G' \homto^? \powstruct{\?B}{\?X}$.
	Well-chosen markings are added to the structures to ensure that the only possible image
	of $s$ (resp. $t$) is $\projHom{1}$ (resp. $\projHom{2}$).
	From there, we get that if there is a homomorphism $\?G' \homto \powstruct{\?B}{\?X}$,
	then $\+G'(s)$ and $\+G'(t)$ cannot be "connected", otherwise $\projHom{1}$ and $\projHom{2}$
	would be connected(ish) too in $\powstruct{\?B}{\?X}$, contradicting the assumption
	that $\?B$ does not have "finite duality".
	The converse property, "ie" that if $\+G'(s)$ and $\+G'(t)$ are not "connected",
	then there is a "homomorphism" $\?G' \homto \powstruct{\?B}{\?X}$,
	is actually trivial and ensured by the construction of $\?G'$.
	This concludes the reduction from "Connectivity in Automatic Graphs" to $\HomAut{\?B}$.
\end{proof}

\subsection{Undecidability of \,$\HomRegAut{\?B}$}
\label{sec:undecidability-homreg}

The reduction to show undecidability of $\HomRegAut{\?B}$
is similar to \Cref{lem:reduction-hom},
but the input problem differs quite a lot. Proofs are given in \Cref{apdx-coro:lowerbound-homreg}.

\decisionproblem{""Regular Unconnectivity in Automatic Graphs""}{
	An "automatic presentation" $\+G$ of a directed graph $\?G$,
	and two elements $s,t \in \Sigma^*$.
}{
	Is there a regular language $L \subseteq \Sigma^*$ 
	such that $s \in L$, $t\not\in L$ and $L$ is a union of "connected components"
	of $\+G$?
}
When $s$ and $t$ satisfy this condition, we say that they are \AP""regularly unconnected"".

\begin{restatable}{lemma}{reductionHomReg}
	\AP\label{lem:reduction-hom-reg}
	Assume that $\sigma$ contains at least one "predicate" of arity at least 2.
	If $\?B$ does not have "finite duality", then there is a "first-order reduction"
	from "Regular Unconnectivity in Automatic Graphs"
	to $\HomRegAut{\marked{\?B}}$.
\end{restatable}

We then prove a lower bound on the complexity of "Regular Unconnectivity in Automatic Graphs".

\begin{lemma}
	\AP\label{lemma:regular-unconnectivity-lowerbound}
	"Regular Unconnectivity in Automatic Graphs" is "RE"-hard.
\end{lemma}

\begin{proof}
	Our reduction starts from the
	"regular reachability problem", which is the problem of, given a Turing machine $\+T$, to decide whether its set of reachable configurations is a regular language.
	It was proven to be undecidable by Barceló, Figueira and Morvan
	in \cite[Lemma~4.2]{BarceloFigueiraMorvan2023SeparatingAutomatic}, even when restricted to what we call here 
	\AP""linear Turing machines"", "ie" machines whose configuration graph is well-founded and
	in which each configuration has at most one predecessor and at most one successor. 
	Note that the terminology of the paper
	differs from ours: it is called
	``\DPfont{reachable regularity problem}'' in \cite{BarceloFigueiraMorvan2023SeparatingAutomatic}, and only undecidability is claimed, but "RE"-hardness trivially holds (the problem is reduced from a
	variation of the halting problem).

	So, we reduce the "regular reachability problem" on "linear Turing machines"
	to "Regular Unconnectivity in Automatic Graphs" in the manner depicted in \Cref{fig:reduction-regreach-to-regunconnect}.

	\begin{figure}
		\centering
		\begin{tikzpicture}
			\input{tikz/conf-graph-wf-RTM}
		\end{tikzpicture}
		\hspace{7em}
		\begin{tikzpicture}
			\input{tikz/reduction-regreach-to-regunconnect}
		\end{tikzpicture}
		\caption{
			\AP\label{fig:reduction-regreach-to-regunconnect}
			Configuration graph of a "linear Turing machine" (left, derived from \cite{BarceloFigueiraMorvan2023SeparatingAutomatic}, licensed under CC BY 4.0)
			and the instance of "Regular Unconnectivity in Automatic Graphs"
			to which it is reduced (right). Colours indicate the different "connected components".
		}
	\end{figure}%
	Given a "linear Turing machine" $\+T$ with
	configuration graph $\tup{V, \+E}$---defined as you would expect,
	see "eg" \cite[Example VII.3.13]{morvan2025thesis}---,
	we reduce it to the "automatic graph" $\+G = \tup{V', \+E'}$ where:
	$V' \defeq V \dcup \{\bullet\}$ and
	$\+E'$ is the union of $\+E$ with the "clique" that
	puts in relation the vertex $\bullet$ and all configurations,
	other than the initial configuration, that have no predecessor.
	We then pick $s$ to be the initial configuration---note that we can assume "wlog" that it has no predecessor---and $t$ to be $\bullet$.

	By construction, $\+G$ is "automatic@@struct" and has exactly two connected components:
	the set of reachable configurations (containing $s$) and its complement (containing $t$).
	Hence, the only union of "connected components" of $\+G$ that contains $s$
	but not $t$ is the set of reachable configurations.
	And hence, $\+T$ is a positive instance of the "regular reachability problem"
	if, and only if, the set of reachable configurations is regular, "ie" $\tup{\+G,s,t}$
	is a positive instance of "Regular Unconnectivity in Automatic Graphs".
\end{proof}

\begin{restatable}{corollary}{lowerboundHomReg}
	\AP\label{coro:lowerbound-homreg}
	If $\?B$ does not have "finite duality", then $\HomRegAut{\?B}$
	is "RE"-hard.
\end{restatable}

\section{Discussion}
\label{sec:dichotomy-discussion}

\subparagraph*{Beyond Automatic Relations.}

In the statement of \Cref{thm:dichotomy-theorem-automatic-structures},
"automatic structures" can be replaced by
$\omega$-tree-automatic structures (see \cite[Definition~XII.1.4]{Blumensath2024MSOModelTheory}), or even
higher-order automatic structures (see \cite[last remark of \S~XII.2]{Blumensath2024MSOModelTheory})
without affecting the validity of the equivalences of the theorem.
The undecidability results obviously carry over to larger classes.
Moreover, to prove decidability, notice that the key lemma of \Cref{sec:first-order-implies-uniformity},
namely \Cref{lemma:finite-duality-uniformly-definable-homomorphisms},
actually deals with $\HomAll{\?B}$ and not $\HomAut{\?B}$.
The decidability follows from the fact that the first-order theory of any
higher-order automatic structure is decidable.

\subparagraph*{Beyond Finite Duality.}

The notion of "duality" has been generalized in graph theory to \emph{restricted dualities},  
in which the quantification over "structures" is restricted to a fixed class of finite graphs.
For instance, Naserasr showed \cite[Theorem 11]{Naserasr2007PlanarDuality} that 
for any finite planar undirected graph $\?G$, $\?G$ is triangle-free
"iff" $\?G$ admits a "homomorphism" to the so-called Clebsch graph.
Letting $\?H_5$ denote the directed graph induced by it---meaning that if $\{u,v\}$ is an edge,
we put an edge from $u$ to $v$ and from $v$ to $u$---it follows that
$\HomAut{\?H_5}$ is decidable when restricted to planar graphs,
even though the full problem is undecidable.

This yields a wide class of problems that are undecidable by the "dichotomy theorem for automatic structures", but admit non-trivial decidable restrictions.
We refer the reader to \cite{NesetrilPOM2012RestrictedDualities} for more details on restricted dualities.

\subparagraph*{Invariance under Graph Isomorphisms.}

Note that given an "automatic presentation" $\+A$ of some "$\sigma$-structure" $\?A$,
the property of whether $\+A \homregto \?B$, where $\?B$ is a "finite $\sigma$-structure",
does not depend only on the "structure" $\?A$, but also on its "presentation" $\+A$.
For instance, the "automatic presentation" of \cite[Example~4.5 \& Figure~4]{BarceloFigueiraMorvan2023SeparatingAutomatic}
has no "regular homomorphism" to $\clique{2}$.
However, it is quite easy to come up with another "presentation@@automatic" that does.

On the other hand, the implication
\itemDTFinDual\ $\Rightarrow$ \itemDTEqual\ proves that if $\?B$ has "finite duality",
then the property of whether $\+A \homregto \?B$ is invariant under "graph isomorphisms",
in the sense that for any "presentations" $\+A_1$ and $\+A_2$ of the "structures"
$\?A_1$ and $\?A_2$, respectively, if $\?A_1$ and $\?A_2$ are "isomorphic", then
$\+A_1 \homregto \?B$ "iff" $\+A_2 \homregto \?B$.
We do not know whether the converse implication holds.

\begin{restatable}{conjecture}{conjInvarianceGraphIsomorphisms}
	\AP\label{conj:invariance-under-graph-isomorphisms}
	For any "finite $\sigma$-structure" $\?B$, $\HomRegAut{\?B}$ is invariant
	under "graph isomorphisms" "iff" $\?B$ has "finite duality".
\end{restatable}

\bibliographystyle{plainurl} 
\bibliography{bibli-abbrev,bibli-dichotomy-theorem-automatic}

\clearpage
\appendix
\section{Supplementary Material for the Preliminaries}

\subsection{Elementary Definitions for Relational Structures}
\label{apdx:relational-structures}

A "$\sigma$-structure" $\?A$ is said to be \AP""finite@@struct"" when 
(1) its domain is finite, (2) for every "predicate" $\+R_{(k)} \in \sigma$, the relation $\+R_{(k)}(\?A)$ is finite, and (3) there exist only finitely many "predicates" $\+R_{(k)} \in \sigma$
"st" $\+R_{(k)}(\?A)$ is non-empty.

A \AP""substructure"" of a "$\sigma$-structure" $\?A$ is another
"$\sigma$-structure" $\?A'$ such that:
\begin{itemize}
	\item the domain $A'$ of $\?A'$ is a subset of $A$,
	\item each "interpretation@@predicate" of a "predicate" in $\?A'$ 
		is a subset of its "interpretation@@predicate" in $\?A$, and
	\item every constant of $\?A$ belongs to $A'$, and the interpretation
		of the constants in both structures coincide. 
\end{itemize}

A \AP""proper substructure"" is a "substructure" for which
at least one of the inclusions in the first two points above
is strict: in other words, such a "substructure" should
miss at least one element, or one hyperedge.

Given a subset $X$ of the domain $A$ of a "$\sigma$-structure" $\?A$,
the \AP""substructure of $\?A$ induced by $X$@induced substructure"" is:
\begin{itemize}
	\item undefined if not all constants of $\?A$ belong to $X$,
	\item otherwise, it is obtained from $\?A$ by restricting
		its domain to $X$, and by intersecting its $k$-ary relations
		with $X^k$.
\end{itemize}

The ""Cartesian product"" \AP$\intro*\prodstruct$ of two structures $\?A$ and $\?B$
is defined over the domain $A\times B$ with the natural "interpretations@@predicate".
For $n\in\N$ we denote by $\intro*\iterstruct{\?A}{n}$ the $n$-fold "Cartesian product" of $\?A$.

Given "structures" $\?A_1, \dotsc, \?A_k$ sharing the same "signature",
the $i$-th projection from the Cartesian product $\?A_1 \prodstruct \dotsc \prodstruct \?A_k$
to $\?A_i$, defined by $\tup{a_1,\dotsc,a_k} \mapsto a_i$ ($i \in \intInt{1,k}$),
is a "homomorphism", and is denoted by \AP$\intro*\projHom{i}$.

\AP An ""isomorphism"" between $\?A$ and $\?B$ is a pair of "homomorphisms"
$f\colon \?A \to \?B$ and $g\colon \?B \to \?A$ such that $f\circ g$ and $g\circ f$
are both the identity on their respective domains.
\AP An ""automorphism"" is an "isomorphism" from a "structure" into itself.
\AP A "homomorphism" is ""strong-onto"" if it is surjective
and if every tuple occurring in the $\+R$-relation of the codomain
is the image of a tuple of the $\+R$-relation of the domain.
\AP Lastly, a "homomorphism" is an ""embedding"" if it is injective.

\subsection{Constructions on Automatic Presentations}
\label{apdx:construction-automatic-presentations}

Let $\+A$ and $\+B$ be "automatic presentations" of some "$\sigma$-structures"
$\?A$ and $\?B$, over alphabets $\Sigma$ and $\Gamma$, respectively.
We define \AP$\+A \intro*\prodpres \+B$ to be the "presentation@@auto"
over the alphabet $(\Sigma \times \Gamma) \dcup (\Sigma \times \{\pad\}) \dcup (\{\pad\} \times 
\Gamma)$  such that:
\begin{align*}
	\domainPres{\+A\prodpres \+B} & \defeq \{u\convol v \mid u \in \domainPres{\+A} \land v \in \domainPres{\+B}\}\\
	\relPres{\+R}{\+A\prodpres \+B} & \defeq
		\{\tup{u_1\convol v_1,\, \dotsc,\, u_k\convol v_k} \mid
		\tup{u_1,\, \dotsc,\, u_k} \in \relPres{\+R}{\+A} \land
		\tup{v_1,\, \dotsc,\, v_k} \in \relPres{\+R}{\+B}
	\}
\end{align*}
for each "predicate" $\+R$ of arity $k$ in $\sigma$.
It is an "automatic presentation" of $\?A \prodstruct \?B$: a "synchronous automaton"
for $\relPres{\+R}{\+A\prodpres \+B}$ checks that each component of its input is of the
form $u \convol v$, and runs in parallel a "synchronous automaton" for $\relPres{\+R}{\+A}$
on the $\Sigma$-components of the input and one for $\relPres{\+R}{\+B}$ on its
$\Gamma$-components, where the padding symbol $\pad$ of the input is projected onto $\pad$
on both sides.
The same construction works for $\domainPres{\+A\prodpres \+B}$.
This shows that if $\?A$ and $\?B$ are "automatic $\sigma$-structures",
then so is $\?A \prodstruct \?B$.

\begin{proposition}
	\label{prop:homreg-prod-finite}
	Let $\?A$, $\?B$ and $\?C$ be "automatic $\sigma$-structures", such that
	$\?B$ and $\?C$ are finite.
	Let $\+A$ (resp. $\+B$ and $\+B'$, resp. $\+C$ and $\+C'$) be an "automatic presentation"
	of $\?A$ (resp. $\?B$, resp. $\?C$).
	Then $\+A \prodpres \+B \homregto \+C$ "iff" $\+A \prodpres \+B' \homregto \+C'$.
\end{proposition}

\begin{proof}
	The proof follows from the following claim, which can be proven
	exactly in the same fashion as \Cref{prop:regular-function-finite-domain}.

	\begin{claim}
		\label{claim:homreg-prod-finite}
		Assuming again that $\?B$ and $\?C$ are finite,
		a function $f\colon\+A \prodpres \+B \to \+C$ is a "regular homomorphism"
		"iff" for every $b\in \domainPres{\+B}$, for every $c\in \domainPres{\+C}$,
		\(\{
			a\in \domainPres{\+A} \mid f(a,b) = c
		\}\)
		is a regular language.\qedhere
	\end{claim}
\end{proof}

In other words, the existence of a "regular homomorphism" does not depend on the
"automatic presentation" of the \emph{finite} "structures" that are involved, but only
on the "structure" they represent.
As a consequence of \Cref{prop:homreg-prod-finite}, we write
\(\+A \prodpres \?B \homregto \?C\) as a synonym for \(\+A \prodpres \+B \homregto \+C\).

\begin{corollary}[Currying]
	\label{coro:homreg-currying}
	Let $\?A$, $\?B$ and $\?C$ be "automatic $\sigma$-structures", such that
	$\?B$ and $\?C$ are finite,
	and let $\+A$ be an "automatic presentation" of $\?A$.
	Then $\+A \prodpres \?B \homregto \?C$ "iff" $\+A \homregto \powstruct{\?C}{\?B}$.
\end{corollary}

\begin{proof}
	This also follows from \Cref{claim:homreg-prod-finite}.
\end{proof}

\subsection{Homomorphisms}
\label{apdx:homomorphisms}

\begin{proposition}
	\AP\label{prop:core-iff-hom-are-auto}
	A "finite $\sigma$-structure" $\?C$ is a "core" if, and only if, every
	homomorphism from $\?C$ to itself is an "automorphism".
\end{proposition}

\begin{proof}
	For the left-to-right implication,
	we let $f\colon \?C \to \?C$ be a "homomorphism".
	Then $f[\?C]$ must be "isomorphic" to $\?C$, otherwise we would obtain
	a strictly smaller retraction. Hence, $f$ is a "strong-onto homomorphism"
	from $\?C$ to itself, and hence is an "automorphism".
	
	Conversely, assuming that any "homomorphism" from $\?C$ to itself is an "automorphism"
	we get in particular that any retraction must be an "automorphism", and
	hence that $\?C$ is "isomorphic" to $\core{\?C}$.
\end{proof}

Given a "$\sigma$-structure" $\?A$ and $a \in A$, we define the \AP""adjacency""
of $a$ in $\?A$ to be the tuple of sets
\AP\phantomintro{\adjacency}
\[
	\reintro*\adjacency{a}{\?A}{\+R}{i} \defeq
		\big\{
			\langle a_1, \dotsc, a_{i-1}, a_{i+1}, \dotsc, a_k \rangle \in A^{k-1}
			\;\big\vert\;
			\langle a_1, \dotsc, a_{i-1}, a, a_{i+1}, \dotsc, a_k \rangle \in \+R(\?A)
		\big\},
\]
when $\+R$ ranges over "predicates" of arity $k$ of $\sigma$ and $i \in \intInt{1,k}$. 
For graphs, the "adjacency" of a vertex corresponds to its set of predecessors and
its set of successors.
The \emph{neighbourhood} refers
to a set of elements, namely the set of elements occurring in the "adjacency".
We say that two elements $a$ and $a'$ are \AP""connected"" when $a'$ belongs to the set obtained from $\set{a}$ by iteratively adding the neighbourhood of the current set until convergence.

\begin{proposition}
	\AP\label{prop:adjacency-core}
	Given a "$\sigma$-structure" $\?B$, if $\?B$ is a "core", then for all $b_1,b_2 \in B$, the following are equivalent:
	\begin{itemize}
		\item $b_1 = b_2$;
		\item $b_1$ and $b_2$ have the same "adjacency";
		\item the adjacency of $b_1$ is included in that of $b_2$.
	\end{itemize}
\end{proposition}

\begin{proof}
	The first two implications are trivial, so we only have to show that
	if the adjacency of $b_1$ is included in that of $b_2$,
	then $b_1=b_2$.
	To do so, consider the "homomorphism" from $\?B$ to itself
	that maps $b_1$ to $b_2$, and all elements of $B \smallsetminus \{b_1\}$
	to themselves. Since we assumed that the "adjacency" of $b_1$
	is included in that of $b_2$, this is indeed a "homomorphism".
	Since $\?B$ is a "core", this "homomorphism" must be an "automorphism",
	and so $b_1 = b_2$.
\end{proof}

\subsection{Proof of Proposition~\ref{prop:idempotent-core-preserves-csp-complexity}}
\label{apdx-prop:idempotent-core-preserves-csp-complexity}

\begin{figure}
	\centering
	\begin{tikzpicture}
		\input{tikz/reduction-idempotent-core-lhs}
	\end{tikzpicture}
	\hspace{1cm}
	\begin{tikzpicture}
		\input{tikz/reduction-idempotent-core-rhs}
	\end{tikzpicture}
	\caption{\AP\label{fig:dichotomy-idempotent-core}
	Reduction from $\,\HomAll{\marked{\?B}}$ to $\,\HomAll{\?B}$ when $\?B$ 
	is the "$2$-transitive tournament":
	we depict on the left-hand side, the "$\extendedSignature{\sigma}{\?B}$-structure"
	$\?A \not\in \HomAll{\marked{\transitiveTournament{2}}}$,
	and on the right-hand side, the "$\sigma$-structure" $\Phi(\?A)
	\not\in \HomAll{\transitiveTournament{2}}$ to which it is reduced.
	The "interpretations@@predicate" of unary predicates in $\?A$
	are described using colours: blue for $0$, red for $1$ and yellow for $2$, as in
	\Cref{fig:examples-graphs}.}
\end{figure}

\idempotentCore*

The non-easy part is to reduce $\HomFin{\marked{\?B}}$ to $\HomFin{\?B}$: only this reduction
requires the assumption that $\?B$ is a core.

\begin{proof}[Proof of \Cref{prop:idempotent-core-preserves-csp-complexity}]
	\case{Reduction from $\,\HomAll{\?B}$ to $\,\HomAll{\marked{\?B}}$.}\\
	We reduce a "$\sigma$-structure" $\?A$ to the
	"$\extendedSignature{\sigma}{\?B}$-structure" $\?A'$ obtained
	from $\?A$ by "interpreting@@predicate" each predicate $\unarypred{b}$ as the empty set.
	Clearly, a function from $A$ to $B$ is a "homomorphism" from $\?A$ to $\?B$
	"iff" it is a "homomorphism" from $\?A'$ to $\marked{\?B}$, proving the correctness
	of the reduction. It is, by definition, "first-order@@reduction".

	\case{Reduction from $\,\HomAll{\marked{\?B}}$ to $\,\HomAll{\?B}$.}
	We first define the reduction $\Phi$ and show its correctness; the fact that it
	is a "first-order reduction" is straightforward.
	We reduce a "$\extendedSignature{\sigma}{\?B}$-structure" $\?A$ to the "$\sigma$-structure"
	$\Phi(\?A)$ illustrated on \Cref{fig:dichotomy-idempotent-core} and defined as follows:
	\begin{itemize}
		\item its underlying universe is the disjoint union $A \dcup B$,
		\item given a "predicate" $\+R$ of arity $k$, its hyperedges are:
		\begin{itemize}
		\item all $\+R$-tuples of $\?A$,
		\item all $\+R$-tuples of $\?B$, and
		\item all $\+R$-tuples $\tup{b_1,\dotsc,b_{i-1}, a_i, b_{i+1},\dotsc,b_k}$
			"st" there exists $b_i$ for which the $\+R$-hyperedge
			$\tup{b_1,\dotsc,b_{i-1}, b_i, b_{i+1},\dotsc,b_k}$
			is in $\+R(\?B)$, and $a_i$ belongs to the "interpretation@@predicate" of 
			$\unarypred{b_i}$ in $\?A$.
		\end{itemize}
	\end{itemize}
	Note that by construction, the "adjacency" of $a \in A$ in $\Phi(\?A)$ is
	the union of its "adjacency" in $\?A$, and the union of the "adjacencies" of
	$b$ in $\?B$ for all $b$ "st" $a \in \unarypred{b}(\?A)$.

	We show that $\?A \in \HomAll{\marked{\?B}}$ "iff" $\Phi(\?A) \in \HomAll{\?B}$.
	So, assume that there exists a "homomorphism" $f\colon \?A \to \marked{\?B}$.
	Then we let $f'\colon A \dcup B \to B$ be defined by $f'(a) = f(a)$ for all $a\in A$ and
	$f'(b) = b$ for all $b \in B$. We claim that $f'$ is a "homomorphism" from
	$\Phi(\?A)$ to $\?B$. Indeed, consider a hyperedge of $\Phi(\?A)$:
	\begin{itemize}
		\item if it is a hyperedge of $\?A$, its image by $f'$ is still a hyperedge
			of $\?B$ since $f$ is a "homomorphism" from $\?A$ to $\marked{\?B}$;
		\item if it is a hyperedge of $\?B$, then its image by $f'$ is itself, and is hence
			a hyperedge of $\?B$;
		\item otherwise, it must be of the form 
			\[\tup{b_1,\dotsc,b_{i-1}, a_i, b_{i+1},\dotsc,b_k}\]
			"st" there exists $b_i$ for which
			$\tup{b_1,\dotsc,b_{i-1}, b_i, b_{i+1},\dotsc,b_k} \in \+R(\?B)$
			and $a_i \in \unarypred{b_i}(\?A)$:
			in this case, its image by $f'$ is 
			\[f'(\tup{b_1,\dotsc,b_{i-1}, a_i, b_{i+1},\dotsc,b_k})
			= \tup{b_1,\dotsc,b_{i-1}, b_i, b_{i+1},\dotsc,b_k} \in \+R(\?B)\]
			since $f'(b) = b$ for all $b\in B$ and $f'(a_i) = f(a_i) = b_i$ since
			$a_i \in \unarypred{b_i}(\?A)$ and $f$ is a "homomorphism" from
			$\?A$ to $\marked{\?B}$.
	\end{itemize}
	And hence, $\Phi(\?A) \homto \?B$.

	Conversely, now let $g\colon \Phi(\?A) \to \?B$ be a "homomorphism".
	Its restriction to $\?B$, namely $\restr{g}{B}$ is a "homomorphism" from $\?B$
	to itself, and since $\?B$ is a "core", it must be an "automorphism" of $\?B$
	by \Cref{prop:core-iff-hom-are-auto}.
	We then define a map $g' \colon A \to B$ by sending
	$a$ to $(\restr{g}{B})^{-1}\circ g(a)$, and claim that it is a "homomorphism"
	from $\?A$ to $\marked{\?B}$. As a matter of fact, it clearly preserves $\+R$-tuples
	for any $\+R$ in $\sigma$, since $g$ and $(\restr{g}{B})^{-1}$ are "homomorphisms".
	We must then show that it preserves all unary "predicates" $\unarypred{b}$, with $b\in B$:
	let $a \in A$ "st" $\unarypred{b}$ holds, "ie" $a \in \unarypred{b}(\?A)$.
	Now, by construction of $\Phi(\?A)$, the "adjacency" of $g(a)$ in $\?B$
	contains that of $g(b)$.
	Since $\?B$ is a "core", it follows by \Cref{prop:adjacency-core} that $g(a) = g(b)$.
	By definition of $g'$, this rewrites as $g'(a) = b$, "ie" $g'(a) \in \unarypred{b}(\marked{\?B})$.
	Therefore, we have built a "homomorphism" from $\?A$ to $\marked{\?B}$.

	Overall, this proves that $\Phi$ is correct.
	It is trivially a "first-order reduction" and moreover,
	it preserves "finiteness@@structure" since $\?B$
	is "finite@@structure".
\end{proof}

\subsection{Critical Obstructions of the 2-Path}
\label{apdx:critical-obst-2path}

We prove here the claim made in \Cref{ex:zigzag-defn}
that each $\zigzag{n}{2}$ ($n\in\N$) is a "critical obstruction" of $\pathGraph{2}$.
We already know that $\zigzag{n}{2}$ is an "obstruction" of $\pathGraph{2}$.
But then notice that to obtain a "proper substructure" of $\zigzag{n}{2}$,
we must either:
\begin{itemize}
	\item remove the edge from $a'_0$ to $a_0$ or the edge from $b_n$ to $b'_n$,
	in which case it admits a "homomorphism" to $\pathGraph{2}$, or
	\item remove any other edge, in which case the resulting "substructure" is not "connected",
	and both parts admit a "homomorphism" to $\pathGraph{2}$.
\end{itemize}
And hence, by \Cref{prop:finite-duality-iff-critical-obstructions}, it follows that
$\pathGraph{2}$ does not have "finite duality".
\section{Supplementary Material for the Main Theorem}
\label{apdx:main-theorem}

\subsection{Regular Homomorphisms}
\label{apdx:regular-homomorphisms}

Given two regular languages $K$ and $L$,%
\footnote{Note that whether $K$ and $L$ share the same alphabet is irrelevant since we can always work on the union of their alphabets.}
a \AP\reintro{regular function} from
$K$ to $L$ is a function $f\colon K \to L$ "st" the relation
$\{\tup{u, f(u)} \mid u \in K\}$
is "automatic@@rel".
A \AP\reintro{regular homomorphism} between two "presentations of automatic $\sigma$-structures"
$\+A$ and $\+B$ is a "regular function" from $\domainPres{\+A}$ to $\domainPres{\+B}$
that defines a "homomorphism" from $\?A$ to $\?B$.\footnote{We use the terminology ``regular''
instead of ``automatic'' simply because ``automatic homomorphism'' sounds somewhat weird.}
We denote by \AP $\+A \intro*\homregto \+B$ the existence of a "regular homomorphism" from
$\+A$ to $\+B$.

\begin{proposition}
	\label{prop:regular-function-finite-domain}
	Let $f\colon K \to L$ be a function, where $L$ is finite.
	Then $f$ is a "regular function" "iff" for every $v \in L$,
	$f^{-1}[v]$ is a regular language.
\end{proposition}

\begin{proof}
	We assume "wlog" that $\Sigma$ has at least two letters, so that \Cref{prop:automatic-first-order} applies.
	For the left-to-right implication, if $f$ is a "regular function",
	then there exists a first-order formula $\phi(x,y)$ "st"
	for all $\tup{u,v} \in K\times L$,
	\[
		\univStructSynchronous{\Sigma}, u, v \FOmodels \phi(x,y)
		\text{ "iff" }
		v = f(u).
	\]
	Then given $v\in L$,
	the formula\footnote{``$y = v$'' is not properly defined in the syntax of "first-order logic" but it is straightforward to come up with a formula expressing this property.}
	\[
		\exists y.\, \phi(x,y) \land y = v
	\]
	describes $\{ u\in K \mid f(u) = v \}$, which is hence regular.

	Conversely, we consider "first-order formulas" $\phi_{v}$ describing each set
	$\{ u\in K \mid f(u) = v \}$, with $v\in L$. Then
	\[
		\phi(x,y) \defeq \bigvee_{v\in L} 
			\big(\phi_{v}(x) \land y = v\big)
	\]
	is a "first-order formula"---since $L$ is finite---describing $f$.
\end{proof}

\subsection{Proof of Proposition~\ref{prop:dichotomy-general-upper-bounds}}
\label{apdx-prop:dichotomy-general-upper-bounds}

\dichotomyGeneralUpperBounds*

Before the proof, we recall a useful theorem.
\begin{proposition}[De Bruijn--Erdős Theorem]
	\AP\label{prop:de-bruijn-erdos}
	Let $\?A$ be an arbitrary "$\sigma$-structure" and $\?B$ a "finite $\sigma$-structure".
	There is a "homomorphism" from $\?A$ to $\?B$ "iff" for every finite "substructure" $\?A'$
	of $\?A$, there is a "homomorphism" from $\?A'$ to $\?B$.%
	\footnote{It is straightforward to note that
	one can replace ``every finite "substructure"'' by
	``every finite "induced substructure"'' in the statement of the theorem.
	The original theorem is about graph colouring, but the generalization is straightforward.}
\end{proposition}

\begin{proof}
	The left-to-right implication is direct.
	We prove the converse by using "Tychonoff's compactness theorem".\footnote{This is a direct adaptation from \cite[\S~``Proof'']{Wikipedia2024DeBruijnErdos}.}
	So, assume that for every finite "substructure" $\?A'$ of $\?A$,
	there is a "homomorphism" from $\?A'$ to $\?B$. 
	Consider the topological space $B^A$, consisting of all functions from $A$ to $B$,
	together with the product topology.\footnote{We equip $B$ with the discrete topology,
	making it compact since $B$ is finite.} By "Tychonoff's compactness theorem",
	$B^A$ is compact. For each finite subset $X$ of $A$, let
	$H_X$ denote the set of all $f \in B^A$ "st" $\restr{f}{X}$ is a "homomorphism"
	from the "substructure" of $\?A$ "induced@@structure" by $X$ to $\?B$.
	Then, each $H_X$ is closed---indeed, whether $f\in B^A$ belongs to $H_X$ only depends
	on finitely many $f(x)$'s---, and moreover the intersection of finitely many
	$H_X$'s, say $H_{X_1} \cap \cdots \cap H_{X_n}$, is non-empty since
	$H_{X_1} \cap \cdots \cap H_{X_n} \supseteq H_{X_1\cup \dotsc \cup X_n}$
	and by assumption $H_{X_1\cup \dotsc \cup X_n}$ is non-empty since $X_1 \cup \cdots \cup X_n$ is finite. Hence, by compactness of $B^A$ and the "finite intersection property", it follows
	that $\bigcap_X H_X$ is non-empty, which means that there is a "homomorphism" from $\?A$ to $\?B$.
\end{proof}


\begin{proof}[Proof of \Cref{prop:dichotomy-general-upper-bounds}]
	\case{$\HomAut{\?B}$ is "coRE".}
	By \Cref{prop:de-bruijn-erdos}, for any "$\sigma$-structure" $\?A$,
	we have $\?A \nothomto \?B$ "iff" there exists a finite "substructure" $\?A'$ of $\?A$
	"st" $\?A' \nothomto \?B$.
	Given a "finite $\sigma$-structure" $\?A'$ and an "automatic $\sigma$-structure" $\?A$,
	it is decidable to test whether $\?A'$ is a "substructure" of $\?A$: indeed, it suffices
	to check, using \Cref{prop:first-order-model-checking-automatic-structures} if
	\[
		\?A
		\FOmodels^{?}
		\Bigl(
			\exists x_1.\; \cdots \exists x_n.\;
			\bigwedge_{\+R_{(k)} \in \sigma} \bigwedge_{\substack{\tup{i_1,\dotsc,i_k} \in \intInt{1,n}^k\\\text{"st" }\tup{a_{i_1},\dotsc,a_{i_k}} \in \+R(\?A')}}
			\+R(x_{i_1},\dotsc,x_{i_k})
			\land \bigwedge_{i \neq j} x_i \neq x_j
		\Bigr),
	\]
	by letting $\{a_1,\dotsc,a_n\} = A'$. 
	Moreover, whether $\?A' \nothomto^? \?B$ is also decidable in "coNP".
	Overall, this provides a co-semi-algorithm for $\HomAut{\?B}$: we enumerate finite
	"$\sigma$-structures" $\?A'$, and test if (1) $\?A'$ is a "substructure" of $\?A$ and if (2)
	$\?A' \nothomto \?B$. And hence $\HomAut{\?B}$ is "coRE".

	\case{$\HomRegAut{\?B}$ is "RE".} This can be proven by defining the notion of
	``$\?B$-automata'': the notion of accepting state is replaced by a partition $\langle C_{b} \rangle_{b \in B}$ of the states.
	The semantics of such an automaton is a partial function $f\colon \Sigma^* \to B$. 
	Given an "automatic presentation" $\+A$, we can then effectively test if $f$ is defined on $\domainPres{\+A}$, and if $f$ defines a "homomorphism" from $\?A$ to $\?B$. If so, $\+A \homregto \?B$. Dually, any "regular homomorphism" $\+A \homregto \?B$ can be described by such a $\?B$-automaton.
	Therefore, $\HomRegAut{\?B}$ is "RE".
\end{proof}
\section{Supplementary Material for Decidability}
\label{apdx:decidability}

\subsection{Trees and Tree Duality}
\label{apdx:tree-duality}

This subsection introduces tools that deal with "tree duality", a generalization of
"finite duality" that we will use in the proof of \Cref{lemma:finite-duality-uniformly-definable-homomorphisms}.

The ""incidence graph"" of a "$\sigma$-structure" $\?A$,
for some "signature" $\sigma$, is the following undirected graph:
\begin{itemize}
	\item its domain is the disjoint union of $A$
		and the hyperedges of $\?A$;
	\item there is an edge between two vertices "iff" one of them
		is a vertex $a$ of $\?A$, and the other is a hyperedge
		$\bar h$ of $\?A$, with the property that $a \in \bar h$.
\end{itemize}

We say that a "$\sigma$-structure" is \AP""strongly acyclic@@struct"" if its "incidence graph" is acyclic, and we define a \AP""$\sigma$-tree"" to be a "$\sigma$-structure" that is both "connected" and "strongly acyclic@@struct".
Do not confuse this notion with the classical notion of directed trees:
every directed tree is a "$\sigma$-tree" for the "graph signature" $\sigma$,
but $\zigzag{n}{2}$---see \Cref{fig:zigzag-graph-hom-T2}---is a "$\sigma$-tree"
while it is not a directed tree. 

We say that a finite "$\sigma$-structure" $\?B$ has \AP""tree duality""
if it admits a (potentially infinite) "dual" consisting only of "$\sigma$-trees".
Somewhat surprisingly, Ne\v{s}et\v{r}il and Tardif showed that
this notion generalizes "finite duality".

\begin{proposition}[{\cite[Theorem~3.1]{NesetrilTardif2000DualityTheorems}}]
	\AP\label{prop:finite-duality-implies-tree-duality}
	If a finite "$\sigma$-structure" $\?B$ has "finite duality",
	then it has "tree duality".%
	\footnote{The statement of \cite[Theorem~3.1]{NesetrilTardif2000DualityTheorems}
	is somewhat cryptic: the relationship with "duals" is given by
	\cite[Lemma~2.5]{NesetrilTardif2000DualityTheorems}.
	We refer the reader to Foniok's Ph.D. thesis for statements that use a terminology closer to ours:
	\cite[Theorem~2.1.12]{Foniok2007PhD} shows that if $\{\?A\}$ is a "dual" of $\?B$,
	then $\?A$ is "homomorphically equivalent" to a tree, "ie" if a "structure" has
	``singleton duality'', then it has "tree duality".
	The generalization to "finite duality" then follows from
	\cite[Theorem~2.4.4]{Foniok2007PhD}.}
\end{proposition}

The converse does not hold.
\begin{proposition}
	\AP\label{prop:2-path-tree-duality}
	The "$2$-path" $\pathGraph{2}$ has "tree duality".
\end{proposition}

\begin{proof}[Proof sketch]
	It can be shown that $\{\zigzag{n}{2} \mid n\in\N\}$ is a "dual" of
	$\pathGraph{2}$. Moreover, each $\zigzag{n}{2}$ ($n\in\N$) is a "$\sigma$-tree".
\end{proof}

Feder and Vardi introduced a construction to decide if a "finite structure" has "tree duality":
given a "$\sigma$-structure" $\?B$, we let \AP$\intro*\FederVardi{\?B}$ be
the $\sigma$-structure whose domain is $\psetp{B}$, and for every $\+R_{(k)} \in \sigma$,
we have $\tup{Y_1, \dotsc, Y_k} \in \+R(\FederVardi{\?B})$ 
(with $Y_1, \dotsc, Y_k \in \psetp{B}$) precisely when 
for every $i\in \intInt{1,k}$, 
for every $b_i \in Y_i$,
there exists $b_j \in Y_j$ for every $j \neq i$ "st" 
$\tup{b_1, \dotsc, b_k} \in \+R(\?B)$.%
\footnote{In fact $\FederVardi{-}$ can
easily be extended to be a functor in the category of $\sigma$-structures, see
"eg" \cite[\S~9.2.2]{NesetrilPOM2012FirstOrderCSPs}.}

By construction, note that $b \mapsto \set{b}$ defines a "homomorphism" from
$\?B$ to $\FederVardi{\?B}$.

\begin{proposition}[{\cite[Theorem 21]{FederVardi1998ComputationalStructure}}]
	\label{prop:charac-Feder-Vardi}
	A "finite $\sigma$-structure" $\?B$ has "tree duality" if, and only if,
	$\FederVardi{\?B} \homto \?B$, or, equivalently, if $\?B$ and $\FederVardi{\?B}$
	are "homomorphically equivalent".
\end{proposition}

Note that $\FederVardi{\FederVardi{\?B}}$ is always "homomorphically equivalent" to
$\FederVardi{\?B}$---see "eg" \cite[\S~9.2.2, Proposition 9.1]{NesetrilPOM2012FirstOrderCSPs}---
and so $\FederVardi{\?B}$ \emph{always} has "tree duality", no matter if $\?B$ has this property.

\begin{example}[{\Cref{ex:zigzag-defn}, continued}]
	\label{ex:Feder-Vardi-P2}
	Using \Cref{prop:charac-Feder-Vardi}, we can prove for instance that
	if a graph $\?G$ has "tree duality", then either it is a DAG, or $\?G$ contains
	a self-loop.

	Indeed, let $v_0 \to v_1 \to \cdots \to v_m$ be a directed cycle in $\?G$,
	with $v_0 = v_m$. Then in $\FederVardi{\?G}$, there is an edge from
	$\{v_0,v_1,\dotsc,v_m\}$ to itself.
	Since $\?G$ has "tree duality", it follows by \Cref{prop:charac-Feder-Vardi}
	that $\FederVardi{\?G} \homto \?G$ and hence $\?G$ also contains a self-loop.
	
	Note that, if $\?G$ contains a self-loop, then actually it is "homomorphically equivalent" to the graph consisting of a single self-loop---that admits $\emptyset$ as a "dual" and hence has
	"tree duality". In other words, what we showed can be rephrased as: any
	non-trivial graph with "tree duality" is a DAG.
	For instance, this implies that $\clique{2}$ does not have
	"tree duality".
	
	\begin{figure}
		\centering
		\begin{tikzpicture}
			\input{tikz/2-path-graph-Feder-Vardi}
		\end{tikzpicture}
		\caption{
			\AP\label{fig:P2-Feder-Vardi}
			The Feder--Vardi construction $\FederVardi{\pathGraph{2}}$
			on the "$2$-path".
		}
	\end{figure}
	On the other hand, we saw in \Cref{ex:zigzag-defn} that $\pathGraph{2}$ does not have
	"finite duality". However, it has "tree duality": indeed, see \Cref{fig:P2-Feder-Vardi}
	and observe that $\FederVardi{\pathGraph{2}} \homto \pathGraph{2}$.
\end{example}

\subsection{Proof of Lemma~\ref{lemma:finite-duality-uniformly-definable-homomorphisms}}
\label{apdx-lemma:finite-duality-uniformly-definable-homomorphisms}

\finiteDualityUniformlyDefinableHomomorphisms*

\begin{proof}
	\case{Converse implication.} Assume that $\HomAll{\?B}$ has
	"uniformly first-order definable homomorphisms", say by $\langle \phi_b(x) \rangle_{b\in B}$.
	Then the conjunction of the properties 
	``every $x$ must satisfy exactly one $\phi_b(x)$ ($b\in B$)'',
	and ``for every "predicate" $\+R$ of arity $k$, for any $\langle x_1,\, \dotsc,\, x_k \rangle$ in $\+R$, there exists $\langle b_1,\, \dotsc,\, b_k \rangle \in \+R(\?B)$ "st"
	each $x_i$ satisfies $\phi_{b_i}$ ($i \in \intInt{1,k}$)'' is a "first-order formula"
	describing all "$\sigma$-structures" of $\HomAll{\?B}$.

	\case{Direct implication.} 
	Let $\?B$ be such that $\HomAll{\?B}$ is "first-order definable", and assume "wlog" that $\?B$ is a "core".
	Notice first that this last assumption can indeed be made "wlog", since 
	$\HomAll{\?B} = \HomAll{\core{\?B}}$ and so $\HomAll{\?B}$ is
	"first-order definable" "iff" $\HomAll{\core{\?B}}$ is.
	Moreover, we claim that if $\HomAll{\core{\?B}}$ has
	"uniformly first-order definable homomorphisms",
	then so does $\HomAll{\?B}$: from formulas $\tup{\phi_{b}(x)}_{b \in \core{B}}$, define formulas $\tup{\phi'_{b}(x)}_{b \in B}$,
	by (1) fixing an arbitrary "embedding" of $\core{\?B}$ into $\?B$,
	and then (2) letting $\phi'_{b}(x) \defeq \phi_{b}(x)$ if $b \in \core{B}$ and $\phi'_b(x) \defeq \bot$ otherwise. It is clear
	that this family of formulas satisfies the desired properties.

	Hence, we now prove the direct implication under the assumption that
	$\?B$ is a "core".
	Given an arbitrary "$\sigma$-structure" $\?A$, we define a function $F\colon A \to \pset{B}$
	by mapping each $a$ to the set of $b$'s ($b \in B$) "st"
	there is a "homomorphism" from $\?A$ to $\?B$ that maps $a$ to $b$.
	\begin{claim}
		\hspace{-.9em}
		\label{claim:finite-duality-uniformly-definable-homomorphisms-homomorphism}
		If $\?A \homto \?B$ then $F$ is a "homomorphism" from $\?A$ to $\FederVardi{\?B}$.
	\end{claim}
	Indeed, since $\?A\homto\?B$, for each $a\in A$ the set
	$F(a)$ is a non-empty subset of $B$---and hence an element of the domain of $\FederVardi{\?B}$.
	We then prove that it is a "homomorphism": let $\+R$ be a "predicate" of arity $l$,
	and let $\langle a_1,\,\dotsc,\,a_l \rangle \in \+R(\?A)$. Then for each $i \in \intInt{1,l}$, for every $b_i \in F(a_i)$, there exists a "homomorphism" $f$ from $\?A$ to $\?B$
	that sends $a_i$ to $b_i$. Then $f(a_j) \in F(a_j)$ for every $j\in \intInt{1,l}$
	and moreover $\langle f(a_1),\,\dotsc,\, f(a_l) \rangle \in \+R(\?B)$.
	Hence, $\langle F(a_1),\,\dotsc,\,F(a_l) \rangle \in \+R(\FederVardi{\?B})$,
	which concludes the proof that $F$ is a "homomorphism" from $\?A$ to $\FederVardi{\?B}$.

	By "Atserias' theorem", since $\HomAll{\?B}$ is "first-order definable",
	then $\?B$ has "finite duality", and in particular it has
	"tree duality" (by \Cref{prop:finite-duality-implies-tree-duality}) and so by
	\Cref{prop:charac-Feder-Vardi}, there exists a "homomorphism"
	$g\colon \FederVardi{\?B} \to \?B$.
	We will now produce "first-order formulas" to describe $g \circ F$.

	If $\HomAll{\?B}$ is "first-order definable", then so is
	$\HomAll{\marked{\?B}}$ by \Cref{prop:idempotent-core-preserves-csp-complexity}, since $\?B$ is a "core".
	So, let $\phi$ be a "first-order formula" over $\extendedSignature{\sigma}{\?B}$
	that describes $\HomAll{\marked{\?B}}$. We let $B = \{b_1,\dotsc,b_k\}$.
	We define a "first-order formula" $\phi^*_i(x_i)$ over $\sigma$,
	where $x_i$ is a fresh variable,
	by substituting $y = x_i$ for each occurrence of $\unarypred{b_i}(y)$ in $\phi$,
	and $\bot$ for $\unarypred{b_j}(y)$ ($j \neq i$).
	Let $\?A$ be a "$\sigma$-structure", $a\in A$ and $i \in \intInt{1,k}$
	and $\?A_{a,i}$ be the
	"$\extendedSignature{\sigma}{\?B}$-structure" obtained by letting 
	$\unarypred{b_i}(\?A_{a,i}) \defeq \{a\}$
	and $\unarypred{b_j}(\?A_{a,i}) \defeq \emptyset$ for all $j \neq i$.

	\begin{claim}\hspace{-.9em}
		\AP\label{claim:finite-duality-uniformly-definable-homomorphisms-new-formulas}
		$\?A_{a,i} \FOmodels \phi$
		"iff" $\langle \?A, a \rangle \FOmodels \phi^*_i(x_i)$.
	\end{claim}
	We prove by induction on "formulas@@FO" $\psi(\bar x)$
	that $\langle\?A_{a,i}, \bar a \rangle \FOmodels \psi$
		"iff" $\langle \?A, \bar a, a \rangle \FOmodels \psi^*_i(x_i)$.
	The base case $\unarypred{b_i}(y)$ is trivial since
	$\langle \?A_{a,i}, a' \rangle \FOmodels \unarypred{b_i}(y)$ "iff"
	$a' = a$ "ie" $\langle \?A, a', a \rangle \FOmodels y = x_i$.
	Similarly, for $\unarypred{b_j}(y)$ ($j \neq i$), we have
	$\langle \?A_{a,i}, a' \rangle \notFOmodels \unarypred{b_j}(y)$
	and so this is equivalent to $\langle \?A, a', a \rangle \FOmodels \bot$.
	The other atomic cases, and inductive cases are trivial.

	\begin{claim}
		\hspace{-.9em}
		\AP\label{claim:finite-duality-uniformly-definable-homomorphisms-formulas}
		There exist "first-order formulas" $\langle \chi_Y(x)\rangle_{Y\in \pset{B}}$,
		that do not depend on $\?A$, "st" for every "$\sigma$-structure" $\?A$
		and for every $a\in A$, we have $\langle \?A,a \rangle \FOmodels \chi_Y(x)$ "iff"
		$F(a) = Y$.
	\end{claim}
	Indeed, given $a\in A$ and $i \in \intInt{1,k}$, there is a "homomorphism" from $\?A$
	to $\?B$ that sends $a$ to $b_i$ "iff" $\?A_{a,i} \FOmodels \phi$, and so by \Cref{claim:finite-duality-uniformly-definable-homomorphisms-new-formulas}, this is equivalent to
	$\langle \?A,a\rangle \FOmodels \phi^*_i(x_i)$. Hence, each $\chi_Y(x)$ can be defined as
	a Boolean combination of the $\phi^*_i(x_i)$'s, after renaming $x_i$ to $x$.

	We can now prove that
	$\HomAll{\?B}$ has "uniformly first-order definable homomorphisms".
	For each $b\in B$, we let $\psi_b(x) \defeq \bigvee_{Y \in g^{-1}[b]} \chi_Y(x)$.
	Now for any "$\sigma$-structure" $\?A$, for any $a\in A$,
	there is at most one $b\in B$ "st" $\langle \?A, a \rangle \FOmodels \psi_b(x)$---indeed,
	there is a unique $Y \in \pset{B}$ (and so at most one $Y \in \psetp{B}$) "st"
	$\langle \?A, a \rangle \FOmodels \chi_Y(x)$ by
	\Cref{claim:finite-duality-uniformly-definable-homomorphisms-formulas}.
	Furthermore, if $\?A \homto \?B$, then for each $a$ there is a unique $b(a) \in B$
	"st" $\langle \?A, a \rangle \FOmodels \psi_{b(a)}(x)$, and moreover $a \mapsto b(a)$
	is a "homomorphism" by 
	\Cref{claim:finite-duality-uniformly-definable-homomorphisms-homomorphism}---that is equal to
	$g \circ F$. And hence, $\HomAll{\?B}$
	has "uniformly first-order definable homomorphisms".
\end{proof}
\subsection{A First-Order Formula for Regular Homomorphisms}
\label{apdx:first-order-formula-homreg}

Using the formulas $\phi^*_b(x)$ from the proof of \Cref{prop:uniformly-first-order-implies-regular}, we can build a "first-order formula"
saying ``every $x$ satisfies exactly one $\phi^*_b(x)$, and moreover
if $\langle x_1,\dotsc,x_k\rangle$ is an $\+R$-tuple then
$\langle b_1,\dotsc,b_k \rangle$ is an $\+R$-tuple, where $b_i$ is the unique element of $B$
"st" $\phi^*_{b_i}(x_i)$ holds''. Each property ``$\langle x_1,\dotsc,x_k\rangle$ is an $\+R$-tuple'' can be expressed using a "first-order formula" expressing the relations
of the "automatic presentation" given as input.

\begin{corollary}[of "Atserias' theorem" and \Cref{lemma:finite-duality-uniformly-definable-homomorphisms}]
	\AP\label{coro:finite-duality-implies-homreg-decidable}
	Let $\?B$ be a finite "$\sigma$-structure" with "finite duality".
	For each "automatic presentation" $\+A$ over an alphabet $\Sigma$ with at least two letters, there exists a "first-order 
	formula" $\phi$, whose size is linear in $\+A$, "st" $\univStructSynchronous{\Sigma} \FOmodels 
	\phi$ "iff" $\+A \homregto \?B$.
\end{corollary}

\section{Supplementary Material for Undecidability}
\label{apdx:dichotomy-undecidability}

\subsection{Proof of Proposition~\ref{prop:currying-hom}}
\label{apdx-prop:currying-hom}

\curryingHom*

\begin{proof}
	Let $\+R$ be a "predicate" of arity $k$, and let
	$\langle a_1,\dotsc,a_k \rangle \in \+R(\?A)$.
	We want to show that $\langle F(a_1),\dotsc,F(a_k) \rangle \in \+R(\powstruct{\?C}{\?B})$:
	for any $\langle b_1,\dotsc,b_k \rangle \in \+R(\?B)$, we have
	\[\langle F(a_1)(b_1), \dotsc, F(a_k)(b_k)\rangle = \langle f(a_1,b_1),\dotsc,f(a_k,b_k) \rangle \in \+R(\?C)\] since $f$ is a "homomorphism" from $\?A\prodstruct \?B$ to $\?C$.
	Hence, $F$ is indeed a "homomorphism" from $\?A$ to $\powstruct{\?C}{\?B}$.

	Dually, if $F$ is a "homomorphism" from $\?A$ to $\powstruct{\?C}{\?B}$,
	we define $f\colon \?A\prodstruct \?B \to \?C$ by $\langle a,b \rangle \mapsto F(a)(b)$,
	and claim that $f$ is a "homomorphism". Indeed, if $\+R$ is a predicate of arity $k$,
	for any $\langle a_1, \dotsc, a_k \rangle \in \+R(\?A)$
	and $\langle b_1, \dotsc, b_k \rangle \in \+R(\?B)$,
	we have $\langle f(a_1,b_1), \dotsc, f(a_k,b_k) \rangle
	= \langle F(a_1)(b_1), \dotsc, F(a_k)(b_k) \rangle$.
	Since $\langle F(a_1), \dotsc, F(a_k) \rangle \in \+R(\powstruct{\?C}{\?B})$
	and $\langle b_1, \dotsc,b_k \rangle \in \+R(\?B)$ 
	it follows that $\langle F(a_1)(b_1), \dotsc, F(a_k)(b_k) \rangle \in \+R(\?C)$.
	Therefore, $f$ is a "homomorphism" from $\?A \prodstruct \?B$ to $\?C$.

	It is then routine to check that the maps $f \mapsto F$ and $F \mapsto f$ defined
	in the two previous paragraphs are mutually inverse bijections.
\end{proof}

\subsection{Links and Finite Duality}

For $n\in\N$, we define the \AP""$n$-link"" $\intro*\link{n}$ to be the "$\sigma$-structure"\footnote{From \cite[\S~2]{LaroseLotenTardif2007CharacterisationFOCSP}.} 
whose domain is $\intInt{0,n}$, and in which every "predicate" $\+R$
of arity $k$ is interpreted as the set of tuples $\langle a_1,\, \dotsc,\, a_k \rangle$
"st" $|a_i-a_j| \leq 1$ for all $i,j \in \intInt{1,k}$. See \Cref{fig:n-link}.
\begin{figure}
	\centering
	\begin{tikzpicture}
		\input{tikz/n-link}
	\end{tikzpicture}
	\caption{\AP\label{fig:n-link}The "$n$-link" $\link{n}$ over the "graph signature".}
\end{figure}
Given a "$\sigma$-structure" $\?B$, we say that $b, b' \in B$ are
\AP""$n$-linked"" if there exists a "homomorphism" from $\link{n}$ to $\?B$
that sends $0$ to $b$ and $n$ to $b'$. We say that $b$ and $b'$ are \AP\reintro{linked} if
they are "$n$-linked" for some $n \in \N$.

Note that the fact that $k \mapsto n-k$
defines an "automorphism" of $\link{n}$ implies that the relation of being "$n$-linked"---%
and \emph{a fortiori} of being "linked"---is symmetric.
Moreover, being "linked" is transitive, but not necessarily reflexive.

\begin{proposition}[{\cite[Theorem 4.7]{LaroseLotenTardif2007CharacterisationFOCSP}}]%
	\AP\label{prop:characterization-finite-duality-path-projections}
	A "finite $\sigma$-structure" $\?B$ has "finite duality" "iff"
	$\projHom{1}$ and $\projHom{2}$ are "linked" in $\powstruct{\?B}{(\iterstruct{\?B}{2})}$.%
	\footnote{Actually \cite[Theorem 4.7]{LaroseLotenTardif2007CharacterisationFOCSP} assumes
	that $\HomFin{\?B}$ is "first-order definable", but this condition
	is equivalent to $\?B$ having "finite duality" by Atserias' result
	\cite[Corollary 4]{Atserias2008DigraphColoring}.}%
\end{proposition}

\subsection{Proof of the Undecidability of $\HomAut{\?B}$}
\label{apdx-lem:lowerbound-hom}

\begin{proposition}[Folklore]
	\AP\label{prop:undecidability-connectivity}
	"Connectivity in Automatic Graphs" is "RE"-complete.
\end{proposition}

\begin{proof}
	We can easily reduce the halting problem to
	"Connectivity in Automatic Graphs".
	First, notice that the configuration graph of a Turing machine is always "automatic@@struct".
	Add then a new node to it, denoted by $\bullet$,
	and put an edge from any halting configuration
	to $\bullet$.
	Then the machine halts
	on the empty word "iff"
	there is a path from the initial configuration to $\bullet$. Since the machine is deterministic,
	this is equivalent to asking that
	the initial configuration and $\bullet$ are "connected".
\end{proof}

\begin{lemma}
	\AP\label{lem:reduction-hom}
	Assume that $\sigma$ contains at least one "predicate" of arity at least 2,
	and let $\?B$ be a "finite $\sigma$-structure".
	If $\?B$ does not have "finite duality", then there is a "first-order reduction" 
	from the complement of "Connectivity in Automatic Graphs" to $\HomAut{\marked{\?B}}$.
\end{lemma}

\begin{figure}
	\centering
	\begin{tikzpicture}[baseline={(0,0)}]
		\input{tikz/graph-to-struct-graph}
	\end{tikzpicture}
	\hspace{2cm}
	\begin{tikzpicture}[baseline={(0,0)}]
		\node[vertex] (0) at (0,0) {};
\node[vertex, right = 2.5 em of 0] (1) {};
\node[vertex, above right = 1em and 2em of 1] (2) {};
\node[vertex, below right = 1em and 2em of 1] (3) {};

\draw[edge, bend right=15] (1) to (2);
\draw[edge, bend right=15] (2) to (1);
\draw[edge, bend right=15] (2) to (3);
\draw[edge, bend right=15] (3) to (2);
\draw[edge, bend right=15] (1) to (3);
\draw[edge, bend right=15] (3) to (1);

\draw[edge, loop, out=60,in=120,looseness=6] (0) to (0);
\draw[edge, loop, out=150,in=210,looseness=6] (1) to (1);
\draw[edge, loop, out=60,in=120,looseness=6] (2) to (2);
\draw[edge, loop, out=-120,in=-60,looseness=6] (3) to (3);
	\end{tikzpicture}
	\caption{
		\AP\label{fig:graph-to-struct}
		A graph $\?G$ (left), and
		the "structure" $\?A$ (right) defined from $\?G$ using the construction done
		in the proof of \Cref{lem:reduction-hom},
		when $\sigma$ consists of a single binary relation.
	}
\end{figure}

\begin{proof}
	Given an instance $\langle \+G, s, t \rangle$ of "Connectivity in Automatic Graphs",
	we first define the $\sigma$-structure $\?A$ with "automatic presentation" $\+A$
	obtained by replacing every vertex by a "$0$-link" and every edge by a "$1$-link".
	Formally, $\?A$ has the same domain as $\?G$, and for any
	"predicate" $\+R \in \sigma$ of arity $k$,
	$\langle g_1,\, \dotsc,\, g_k \rangle \in \+R(\?A)$ "iff"
	$\{g_1, \dotsc, g_k\} \subseteq \{g,g'\}$ for some $g,g' \in G$ "st"
	$g = g'$ or there is an edge from $g$ to $g'$ in $\?G$.
	See \Cref{fig:graph-to-struct}.

	\begin{claim}
		\AP\label{claim:reduction-hom-from-graph-to-link}
		$\+G(s)$ and $\+G(t)$ are "connected" "iff"
		$\+A(s)$ and $\+A(t)$ are "linked".
	\end{claim}
	The left-to-right implication is easy: any
	path of size $n$ from $\+G(s)$ to $\+G(t)$ in $\?G$ yields a "homomorphism" from $\link{n}$ to $\?A$
	sending $0$ to $\+A(s)$ and $n$ to $\+A(t)$.
	Conversely, if two elements $a$ and $a'$ of $\?A$ are "$1$-linked", 
	then pick a "predicate" $\+R \in \sigma$ of arity at least 2.
	Then $\langle a,\, \dotsc,\, a,\, a' \rangle \in \+R(\?A)$,
	and so by definition of $\?A$ either $a = a'$, or there is an edge from $a$ to $a'$
	or from $a'$ to $a$ in $\?G$.\footnote{Note that the proof of this claim
	is the only part of the proof of \Cref{lem:reduction-hom} that requires
	the assumption that $\sigma$ contains at least one "predicate" of arity at least 2.}

	We then consider the "automatic $\sigma$-structure" $\?A\prodstruct \iterstruct{\?B}{2}$,
	and extend it to a
	"$\extendedSignature{\sigma}{\?B}$-structure" \AP\(\intro*\ConstrUndecHom{(\?A\prodstruct \iterstruct{\?B}{2})}\)
	in which for each $b_0 \in B$,
	we "interpret@@predicate" the unary "predicate" $\unarypred{b_0}$ as
	\[
		\big\{
			\langle \+A(s),\, b_0,\, b\rangle \;\big\vert\; b \in B
			\big\}
		\cup
		\big\{
			\langle \+A(t),\, b,\, b_0\rangle \;\big\vert\; b \in B
		\big\}
	\]
	To construct an "automatic presentation" for this structure, see \Cref{apdx:construction-automatic-presentations}.
	\begin{claim}
		\AP\label{claim:reduction-hom-direct}
		If there is a "homomorphism" from
		$\ConstrUndecHom{(\?A\prodstruct \iterstruct{\?B}{2})}$
		to $\marked{\?B}$,
		then $\+G(s)$ and $\+G(t)$ are not "connected" in $\?G$.
	\end{claim}
	Let $f\colon \ConstrUndecHom{(\?A\prodstruct \iterstruct{\?B}{2})} \to \marked{\?B}$
	be a "homomorphism".\footnote{Recall that both sides are
	"$\extendedSignature{\sigma}{\?B}$-structures".}
	It induces a "homomorphism"
	\[
		\overline f\colon \?A\prodstruct \iterstruct{\?B}{2} \to \?B
	\]
	between "$\sigma$-structures", and by currying (\Cref{prop:currying-hom}),
	$\overline f$ can be seen as a "homomorphism"
	\[
		F\colon \?A \to \powstruct{\?B}{(\iterstruct{\?B}{2})}.
	\]
	Note moreover that since $\overline f$ comes from a "homomorphism" between
	"$\extendedSignature{\sigma}{\?B}$-structures", we must have  
	$f(\+A(s),\, b,\, b') = b$
	and $f(\+A(t),\, b,\, b') = b'$ for all $b,b' \in B$.
	This implies that $F(\+A(s)) = \projHom{1}$ and $F(\+A(t)) = \projHom{2}$.
	
	Assume, for the sake of contradiction, that
	$\+G(s)$ and $\+G(t)$ are "connected", and hence by
	\Cref{claim:reduction-hom-from-graph-to-link}
	there is some $n \in \N$
	"st" there is a "homomorphism" $g\colon \link{n} \to \?A$
	with $g(0) = \+A(s)$ and $g(n) = \+A(t)$.
	Then by composition, we obtain a "homomorphism"
	\[
		F \circ g\colon
		\link{n} \to \powstruct{\?B}{(\iterstruct{\?B}{2})},
 	\]	
	which sends $0$ to $F(g(0)) = F(\+A(s)) = \projHom{1}$
	and sends $n$ to $F(g(n)) = F(\+A(t)) = \projHom{2}$.
	So, by \Cref{prop:characterization-finite-duality-path-projections},
	$\?B$ would have "finite duality", which is a contradiction.
	Hence, $\+A(s)$ and $\+A(t)$ are not "linked",
	and so by \Cref{claim:reduction-hom-from-graph-to-link}, $\+G(s)$ and $\+G(t)$
	are not "connected".

	\begin{claim}
		\AP\label{claim:reduction-hom-converse}
		If $\+G(s)$ and $\+G(t)$ are not "connected" in $\?G$,
		then there is a "homomorphism" from
		$\ConstrUndecHom{(\?A\prodstruct \iterstruct{\?B}{2})}$ to $\marked{\?B}$.
	\end{claim}
	We define a "homomorphism" $f\colon \ConstrUndecHom{(\?A\prodstruct \iterstruct{\?B}{2})} \to \marked{\?B}$ by:
	\[
		f(a, b, b') \defeq \begin{cases*}
			\;b & \text{ if $\+A(s)$ and $a$ are "linked",} \\
			\;b' & \text{ otherwise.}
		\end{cases*}
	\]
	We show that this is indeed a "homomorphism": for any "predicate" $\+R$
	of arity $k$ in $\sigma$, if
	\[
		\langle a_1,\, b_1,\, b'_1 \rangle,\;
		\langle a_2,\, b_2,\, b'_2 \rangle,\;
		\dotsc,\;
		\langle a_k,\, b_k,\, b'_k \rangle
	\]
	are all $\+R$-tuples of $\ConstrUndecHom{(\?A\prodstruct \iterstruct{\?B}{2})}$,
	then by definition of $\?A$, we have that either (1) all $a_i$'s are equal,
	or (2) $\{a_1,\, \dotsc,\, a_k\} = \{a,a'\}$ for some $a \neq a' \in A$
	and there is an edge from $a$ to $a'$ or from $a'$ to $a$ in $\?G$.
	In both cases, it follows that $\+A(s)$ and $a_i$ are "linked"
	"iff" $\+A(s)$ and $a_j$ are "linked", for all $i,j\in \intInt{1,k}$.
	Hence, either:
	\begin{itemize}
		\item $f(a_i,\, b_i,\, b'_i) = b_i$ for all $i\in \intInt{1,k}$
			(if all $a_i$'s are "linked" to $\+A(s)$), or
		\item $f(a_i,\, b_i,\, b'_i) = b'_i$ for all $i\in \intInt{1,k}$ (otherwise).
	\end{itemize}
	In both cases, we get that
	\[
		\Big\langle
			f(a_1,\, b_1,\, b'_1),\;
			f(a_2,\, b_2,\, b'_2),\;
			\dotsc,\;
			f(a_k,\, b_k,\, b'_k)
		\Big\rangle
		\in \+R(\?B).
	\]
	We also need to show that this map preserves the new unary "predicates" of
	$\extendedSignature{\sigma}{\?B}$: this follows from---and is in fact equivalent to---the
	fact that $\+A(s)$ and $\+A(t)$ are not "linked" by \Cref{claim:reduction-hom-from-graph-to-link}:
	indeed, the currying of $f$ behaves like $\projHom{1}$ on the "connected component"
	of $\+A(s)$ and like $\projHom{2}$ on its complement.
	Overall, this proves that $\ConstrUndecHom{(\?A\prodstruct \iterstruct{\?B}{2})} \homto \marked{\?B}$.

	Putting the claims together,
	we get that the reduction is correct.
	Lastly, note that it is a "first-order reduction":
	clearly, one can go from $\?G$ to $\?A$ via a (one-dimensional) "first-order reduction",
	and then from $\?A$ to $\?A\prodstruct \iterstruct{\?B}{2}$
	via a (multi-dimensional) "first-order reduction" since $\?B$ is finite,
	and lastly 
	$\ConstrUndecHom{(\?A\prodstruct \iterstruct{\?B}{2})}$ can be obtained
	by a "first-order reduction" from the latter "structure" since
	"first-order logic" can test equality with a fixed element.
\end{proof}

By \Cref{prop:undecidability-connectivity}, the complement of "Connectivity in Automatic Graphs"
is "coRE"-complete, and assuming that $\sigma$ contains at least one "predicate" of arity at least 2,
it reduces by \Cref{lem:reduction-hom} to any problem $\HomAut{\marked{\?B}}$ when $\?B$ does not have "finite duality". In turn, by \Cref{prop:idempotent-core-preserves-csp-complexity}, it reduces to
$\HomAut{\?B}$, which is thus "coRE"-hard. It remains to deal with "signatures" consisting of only
unary "predicates".\footnote{It is not clear to us whether this case was properly handled in
\cite{LaroseTesson2009UniversalAlgebraCSP}.}

\begin{proposition}
	\AP\label{prop:finite-duality-unary-predicates}
	If $\sigma$ only consists of unary "predicates", then all "$\sigma$-structures"
	have "finite duality".
\end{proposition}	

\begin{proof}
	Fix a $\sigma$-structure $\?B$. We define the \AP""unary type""
	$\intro*\unaryType{b}{\?B}$ of $b \in \?B$
	to be the set of "predicates" $\+P$ "st" $b \in \+P(\?B)$.
	
	Given $\tau \subseteq \sigma$, define \AP$\intro*\structOfUnaryType{\tau}$
	to be the "$\sigma$-structure"
	consisting of a single element $*$, and "st" $* \in \+P(\?1_\tau)$ "iff"
	$\+P \in \tau$.
	We say that $\tau$ is \AP""obstructing@@unarytype"" if
	$\tau \not\subseteq \unaryType{b}{\?B}$ for all $b \in \?B$.

	\begin{claim}
		\AP\label{claim:finite-duality-unary-predicates-direct}
		If $\tau$ is "obstructing@@unarytype",
		then $\structOfUnaryType{\tau} \nothomto \?B$.
	\end{claim}
	We prove the result by contraposition.
	Any "homomorphism" from $\structOfUnaryType{\tau}$ to $\?B$
	should send $*$ to some element $b$ of $\?B$
	"st" $b \in \+P(\?B)$ for all $\+P \in \tau$, and
	hence $\tau \subseteq \unaryType{b}{\?B}$.

	\begin{claim}
		\AP\label{claim:finite-duality-unary-predicates-converse}
		If $\?A \nothomto \?B$ then there exists an "obstructing@@unarytype"
		$\tau \subseteq \sigma$ "st" $\structOfUnaryType{\tau} \homto \?A$.
	\end{claim}
	We define a partial homomorphism $f$ from $A$ to $B$,
	by sending $a \in A$ to any $b \in B$ "st" the "unary type" of $a$
	is included in the "unary type" of $b$. This is clearly a (partial) "homomorphism",
	and so since $\?A \nothomto \?B$, it follows that it must be partial,
	"ie" that there is some element $a \in \?A$ "st" $\unaryType{a}{\?A} \not\subseteq
	\unaryType{b}{\?B}$ for every $b \in B$. It follows that $\unaryType{a}{\?A}$
	is "obstructing@@unarytype". Since $\structOfUnaryType{\unaryType{a}{\?A}} \homto \?A$
	"via" $* \mapsto a$, the conclusion follows.

	Putting the claims together, we get that
	\[
		\big\{\;
			\structOfUnaryType{\tau}
			\;\big\vert\;
			\text{ $\tau \subseteq \sigma$ is "obstructing@@unarytype"} 
		\;\big\}
	\]
	is a finite "dual" for $\?B$.
\end{proof}

\lowerboundHom*

\begin{proof}
	By \Cref{prop:finite-duality-unary-predicates}, since $\?B$ does not have "finite duality",
	then $\sigma$ has at least one "predicate" of arity at least 2.
	Since $\HomAut{\?B} = \HomAut{\core{\?B}}$, we assume "wlog" that $\?B$ is a "core",
	and so the conclusion follows from \Cref{prop:idempotent-core-preserves-csp-complexity,prop:undecidability-connectivity,lem:reduction-hom}.
\end{proof}

\subsection{Proof of the Undecidability of $\HomRegAut{\?B}$}
\label{apdx-coro:lowerbound-homreg}

\reductionHomReg*
\begin{proof}
	Given an instance $\langle \+G, s, t \rangle$ of "Regular Unconnectivity in Automatic Graphs",
	we first define the $\sigma$-structure $\?A$ with "automatic presentation" $\+A$
	obtained by replacing every edge by a "$1$-link", as in \Cref{lem:reduction-hom}.

	\begin{claim}
		\AP\label{claim:reduction-homreg-from-graph-to-link}
		$\+G(s)$ and $\+G(t)$ are "regularly unconnected" "iff"
		there is a regular language $L \subseteq \Sigma^*$ "st" $\+A(s)\in L$ and $\+A(t) \not\in L$,
		and $L$ is a union of equivalence classes of $\domainPres{\+A}$
		under ``being "linked"''.
	\end{claim}

	In the statement above, we talk about ``being "linked"''
	as an equivalence relation: while this relation is not reflexive in general, it is over the
	structure $\?A$ by \Cref{claim:reduction-hom-from-graph-to-link}, by reflexivity of ``being "connected"'' in $\+G$.
	The proof is similar to that of \Cref{claim:reduction-hom-from-graph-to-link}.
	
	Then again, we reduce the instance $\langle \+G, s, t \rangle$
	to an "automatic presentation" of \(\ConstrUndecHom{(\?A\prodstruct \iterstruct{\?B}{2})}\),
	as in \Cref{lem:reduction-hom}.
	\begin{claim}
		\AP\label{claim:reduction-homreg-direct}
		If $\ConstrUndecHom{(\+A\prodpres \iterstruct{\?B}{2})} \homregto \marked{\?B}$,
		then $\+G(s)$ and $\+G(t)$ are "regularly unconnected" in $\?G$.
	\end{claim}
	
	Let \(f\colon \ConstrUndecHom{(\+A\prodpres \iterstruct{\?B}{2})} \to \marked{\?B}\)
	be a "regular homomorphism".
	By currying---see \Cref{coro:homreg-currying}---of the underlying "homomorphism"
	between "\(\sigma\)-structures", we obtain a "regular homomorphism"
	\[
		F\colon \+A \to \powstruct{\?B}{(\iterstruct{\?B}{2})}.
	\]
	Moreover, using the "predicates" \(\unarypred{b}\), \(b \in B\),
	we get that $F(\+A(s)) = \projHom{1}$ and $F(\+A(t)) = \projHom{2}$.

	We then define \[\+X \defeq \{g \in \powstruct{\?B}{(\iterstruct{\?B}{2})} \mid \text{ $g$ and $\projHom{1}$ are "linked" or $g = \projHom{1}$}\}.\]
	We claim that ${F}^{-1}[\+X]$ witnesses the fact that
	$\+G(s)$ and $\+G(t)$ are "regularly unconnected".
	First, $\projHom{1} \in \+X$ so $\+A(s) \in {F}^{-1}[\+X]$.
	Since $\?B$ does not have "finite duality", by \Cref{prop:characterization-finite-duality-path-projections}, $\projHom{2} \not\in \+X$
	and so $\+A(t) \not\in {F}^{-1}[\+X]$.
	Then, ${F}^{-1}[\+X]$ is regular since $F$ is a "regular homomorphism". Finally, ${F}^{-1}[\+X]$ is a union of
	equivalence classes of $\domainPres{\+A}$ under ``being "linked"''.\footnote{Indeed,
	if $c_1, c_2 \in \?C$ are "linked" in some "structure" $\?C$ and if $f\colon \?C \to \?D$ is a "homomorphism", then $f(c_1)$ and $f(c_2)$ are "linked" in $\?D$.}
	Hence, by \Cref{claim:reduction-homreg-from-graph-to-link}, $\+G(s)$ and $\+G(t)$ are "regularly unconnected".

	\begin{claim}
		\AP\label{claim:reduction-homreg-converse}
		If $\+G(s)$ and $\+G(t)$ are "regularly unconnected" in $\?G$,
		then $\ConstrUndecHom{(\+A\prodpres \iterstruct{\?B}{2})} \homregto \marked{\?B}$.
	\end{claim}

	Since $\+G(s)$ and $\+G(t)$ are "regularly unconnected" in $\?G$,
	by \Cref{claim:reduction-homreg-from-graph-to-link} there is a regular language $L \subseteq \Sigma^*$ "st" $\+A(s)\in L$ and $\+A(t) \not\in L$,
	and $L$ is a union of equivalence classes of $\domainPres{\+A}$
	under ``being "linked"''.
	We define a function $f\colon \domainPres{\+A}\times B^2 \to B$ by 
	\[
		f(a, b, b') \defeq \begin{cases*}
			\;b & \text{ if $a \in L$,} \\
			\;b' & \text{ otherwise,}
		\end{cases*}
	\]
	and we claim that $f$ is a "regular homomorphism" from
	\(\ConstrUndecHom{(\+A\prodpres \iterstruct{\?B}{2})}\) to \(\marked{\?B}\).
	The proof that it is a "homomorphism" is similar to that of \Cref{claim:reduction-hom-converse}:
	in particular, we use the fact that $\+G(s)$ and $\+G(t)$ are not "connected" in $\?G$,
	which is a consequence of the fact that they are "regularly unconnected".
	"Regularity@@hom" follows from the regularity of $L$. 
	Hence, $\ConstrUndecHom{(\+A\prodpres \iterstruct{\?B}{2})} \homregto \marked{\?B}$.

	Putting the claims together,
	we get that the reduction is correct.
	Lastly, note that this is a first-order reduction for the same reason
	as \Cref{lem:reduction-hom}.
\end{proof}

\lowerboundHomReg*

\begin{proof}
	Recall that $\HomRegAut{\?B} = \HomRegAut{\core{\?B}}$, so we assume "wlog"
	that $\?B$ is a "core".
	By \Cref{lemma:regular-unconnectivity-lowerbound},
	"Regular Unconnectivity in Automatic Graphs" is "RE"-hard.
	Then by \Cref{prop:finite-duality-unary-predicates}, since $\?B$ does not have "finite
	duality", $\sigma$ does not consist only of unary "predicates",
	and hence by \Cref{lem:reduction-hom-reg}, we get
	a reduction from "Regular Unconnectivity in Automatic Graphs" to
	$\HomRegAut{\marked{\?B}}$, which in turn
	reduces to $\HomRegAut{\?B}$ by \Cref{prop:idempotent-core-preserves-csp-complexity}
	since $\?B$ was assumed to be a "core".
	Indeed, "first-order reductions" preserve "regularity@@hom", by 
	\Cref{prop:automatic-first-order}.
\end{proof}

\end{document}